\documentclass[11pt]{article}
\usepackage{amsfonts}
\usepackage{amssymb}
\usepackage{graphicx}
\usepackage{amsmath}
\usepackage{amsthm}
\usepackage{color}
\usepackage{multirow}
\usepackage{url}
\usepackage{float}
\usepackage{booktabs} 
\usepackage{makecell}
\usepackage{listings}
\usepackage{subcaption}
\usepackage{algorithm}
\usepackage{enumerate}
\usepackage{tikz}

\usepackage{algpseudocode}
\lstnewenvironment{Rcode}{\lstset{language=R}}{}

\newtheorem{theorem}{Theorem}[section]

\newtheorem{lemma}[theorem]{Lemma}

\newtheorem{remark}[theorem]{Remark}

\definecolor{battleshipgrey}{rgb}{0.52, 0.52, 0.51}
\definecolor{navyblue}{rgb}{0.0, 0.0, 0.5}
\definecolor{arsenic}{rgb}{0.23, 0.27, 0.29}
\definecolor{oldmauve}{rgb}{0.4, 0.19, 0.28}
\usepackage[colorlinks=TRUE]{hyperref}
\hypersetup{
	colorlinks=true,
	linkcolor=navyblue,
	filecolor=blue,      
	urlcolor=oldmauve,
	citecolor=navyblue,
	pdftitle={Overleaf Example},
	pdfpagemode=FullScreen,
}
\usepackage{natbib}
\setcitestyle{authoryear}

\begin{document}
		
	\title{\bf kNN estimation in semi-functional partial linear regression with missing responses at random}
	\author{Germ\'{a}n Aneiros{$^{a,b}$} \hspace{2pt} Silvia Novo{$^c$}\footnote{Corresponding author email address: \href{snovo@est-econ.uc3m.es}{snovo@est-econ.uc3m.es}}  \\		
		{\normalsize $^a$ Grupo de Investigación MODES, Departamento de Matemáticas, Universidade da Coruña, A Coruña, Spain}\\
		{\normalsize $^b$ Centro de Investigaci\'on en Tecnolog\'ias de la Informaci\'on y las Comunicaciones (CITIC), A Coruña, Spain}\\{\normalsize $^c$ Departamento de Estad\'istica, Universidad Carlos III de Madrid, Madrid, Spain}\\
			}
	
	\date{}
	\maketitle
	\begin{abstract} 
This paper considers a partial linear regression model with scalar response missing at random, one finite-dimensional covariate (a vector, $X$) and one infinite-dimensional covariate (a functional variable, $\mathcal{X}$). While the effect of $X$ on the response is linear, the effect of $\mathcal{X}$ is nonparametric. Three $k$NN-based estimators are proposed for both the vector parameter and the nonparametric operator, and some first asymptotic results are obtained.  
	\end{abstract}
	
	
\section{Introduction}
This paper focuses on the Semi-Functional Partial Linear (SFPL) regression model
\begin{equation}
	\label{eq_model1}
	Y=X^{\top}\pmb{\beta}+m\left(\mathcal{X}\right)+\varepsilon ,
\end{equation}
where $\pmb{\beta}=\left(\beta_{1},\dots,\beta_{p}\right)^{\top}\in\mathbb{R}^p$ is a vector of unknown parameters, $m(\cdot)$ is an unknown real-valued operator, the response variable ($Y$) takes values in $\mathbb{R}$, $X$ is an explanatory variable taking values in $\mathbb{R}^p$, $\mathcal{X}$ is another explanatory variable but of functional nature, and the random error, $\varepsilon$, satisfies $\mathbb{E}\left(\varepsilon | X, \mathcal{X}\right)=0.$ Throughout all this paper, we will assume that $\mathcal{X}$ is valued in $S_\mathcal{F} \subset \mathcal{F}$, where $\mathcal{F}$ is some abstract infinite-dimensional semi-metric space whose associated semi-metric is denoted by $d(\cdot,\cdot)$.

Model (\ref{eq_model1}) was introduced in \cite{anev06} where, in a context of independent data, some first results were obtained, including asymptotic normality of an estimator of $\pmb{\beta}$ and rates of almost sure uniform convergence of an estimator of $m(\cdot).$ Since the publication of \cite{anev06}, this model has been studied extensively. Some recent papers include bandwidth selection (\citealt{sha14}), variable selection (\citealt{ane15}), robust estimation (\citealt{boev17}), quantile regression (\citealt{din18}) and testing of hypotheses (\citealt{zhuz19}), among others. Extensions of model (\ref{eq_model1}) were also considered in the statistical literature, as for instance the cases of: dependent data (\citealt{anev08}), covariates with measurement error (\citealt{zhu20}), missing at random responses (\citealt{lin19}) and censored responses (\citealt{wan26}). In addition, applications in diverse fields as Chemometrics (\citealt{anev06}), Energy (\citealt{vil12}) and Medicine (\citealt{wan26}) support the usefulness in practice of model (\ref{eq_model1}). For a comprehensive foundation in Functional Data Analysis (FDA), the interested reader can see the pioneering monographs by \cite{rams05} and \cite{ferv06}, which provide introduction to parametric and nonparametric methods, respectively; other interesting monographs are \cite{hork12}, \cite{hsie15} and \cite{kokr17}. For recent advances on FDA, see the special issues presented in \cite{koko17} and \cite{anehhv22}.

Missing data is a common challenge in data statistical analysis that occurs when some values for certain variables are not available. This issue typically arises from source non-response, system errors, or data corruption during collection. If left unaddressed, missing values can severely bias statistical estimates, reduce the power of a model, and lead to inaccurate conclusions. Consequently, identifying the underlying missingness mechanism is a critical first step before applying handling techniques like deletion or imputation. One of such mechanisms is the missing at random (MAR) one. See \cite{litr87} for a first monograph on missing data. 
In this paper, 
we develop inference on the SFPL model (\ref{eq_model1}) where the responses are MAR. Formally, we assume that $(Y_i,\delta_i,X_i,\mathcal{X}_i), \ i=1,\ldots n,$ are independent and identically distributed, where $(Y_i,X_i,\mathcal{X}_i)$ comes from model (\ref{eq_model1}),
\begin{equation}
\delta_i=0 \mbox{ if } Y_i \mbox{ is missing and } \delta_i=1 \mbox{ otherwise}, \mbox{ and } P(\delta_i=1 | Y_i,X_i,\mathcal{X}_i)=P(\delta_i=1 | X_i,\mathcal{X}_i). \label{MAR}
\end{equation}
Therefore, the MAR condition states that $Y$ and $\delta$ are conditional-independent, given $(X,\mathcal{X})$. 
\cite{wan04} and \cite{wan07} were the first papers dealing with (non-functional) partial linear regression models with MAR responses. The aim of \cite{wan04} was inference on the unconditional mean of $Y$ while \cite{wan07} focused on inference on $\pmb{\beta}$ and $m(\cdot).$ In both cases, the covariate $\mathcal{X}$ was finite dimensional (a vector), and the developed estimators were based on imputation, semiparametric regression surrogate and inverse marginal probability weighted approaches. \cite{lin19} considered the case where $\mathcal{X}$ in the SFPL-MAR model (\ref{eq_model1})-(\ref{MAR}) was infinite dimensional (a functional data), and proposed and studied estimators based on imputation. In these three papers, kernel regression techniques were used. See \cite{fers13} for a first paper on functional nonparametric models with MAR responses, and \cite{crah19} and \cite{febg19} for the case of linear models instead of nonparametric ones.

This paper aims to extend the results in \cite{wan07} in two ways. On the one hand, to address the case where the covariate $\mathcal{X}$ is of functional nature. On the other hand, to deal with (weighted) $k$NN ($k$-Nearest Neighbors) regression. Specifically, three estimators based on $k$NN ideas will be proposed for each unknown component, $\pmb{\beta}$ and $m(\cdot)$, in the regression function of the SFPL-MAR model (\ref{eq_model1})-(\ref{MAR}). Then, some asymptotic properties will be obtained: normality for the estimators of $\pmb{\beta}$ and almost sure uniform rates of convergence for the estimators of $m(\cdot)$. As far as we know, this is the first paper in literature that addresses inference on the SFPL-MAR model based on $k$NN regression.

The rest of this paper is organized as follows. In Section \ref{estim} we present three classes of $k$NN-based estimators for $\pmb{\beta}$ and $m(\cdot)$: imputation, semiparametric regression surrogate and  inverse marginal probability weighted estimators, and show some asymptotic results. 
All proofs are given in Appendix \ref{app}, which also contains some technical lemmas.

\section{Estimators and main results}\label{estim}

Kernel regression and (weighted) $k$NN regression are two popular nonparametric techniques for estimating a regression function, $r(\chi),$ from a sample, $\{(Y_i,\mathcal{X}_i)\},$ verifying $Y_i=r(X_i)+\varepsilon_i$. Both techniques estimate values $r(\chi)$ by weighted average of responses $Y_i,$ in such a way the closer $\chi$ and $\mathcal{X}_i$ are, the higher the weight for $Y_i$ is. The difference between these kind of estimators is that, while the kernel estimator averages all responses $Y_i$ whose $X_i$ lies within a ball centered at $\chi$ and with radius $h>0$ (the bandwidth), the $k$NN one averages a fixed number of $Y_i$, $k$ (corresponding to the $k$-Nearest Neighbors to $\chi$). 
Specifically, the weights used in kernel regression can be written as 
	\begin{equation}
		\omega_{h}(\chi,\mathcal{X}_i)=\frac{K\left(h^{-1}d(\mathcal{X}_i,\chi)\right)}{\sum_{i=1}^nK\left(h^{-1}d(\mathcal{X}_i,\chi)\right)}, \label{pesos-kernel}
	\end{equation}
while the ones considered in $k$NN regression are
\begin{equation}
	\omega_{k}(\chi,\mathcal{X}_i)=\frac{K\left(H_{k,\chi}^{-1}d\left(\mathcal{X}_i,\chi\right)\right)}{\sum_{i=1}^nK\left(H_{k,\chi}^{-1}d\left(\mathcal{X}_i,\chi\right)\right)},
	\label{pesos-kNN}
\end{equation}
where
$$
H_{k,\chi}=\min\left\{h\in \mathbb{R}^+ \mbox{\text{ such that }} \sum_{i=1}^n1_{B(\chi,h)}(\mathcal{X}_i)=k\right\}, 
$$
with $B(\chi,h)=\left\{z\in \mathcal{F}:d\left(\chi,z\right) \leq h\right\}.$ $K(\cdot)$ is a nonnegative function (known as kernel function) verifying certain conditions. 

Note that, from the practical point of view, the $k$NN estimators have some advantages over the kernel ones. For instance,
	although the number of neighbors, $k$, is fixed, the bandwidth $H_{k,\chi}$ varies with $\chi$, providing the local-adaptive property of $k$NN estimators (allowing adaptation to heterogeneous designs). In addition, the selection of the smoothing parameter, $k,$ has a lower computational cost than the selection of the bandwidth $h$, since $k$ takes values in $\{1,2,\dots,n\}$ (finite discrete set) while $h$ takes them in an interval (continuous set). The drawback of $k$NN estimators arises when one wants to validate their behavior through theoretical results; the fact that $H_{k,\chi}$ is a random variable makes obtaining such results difficult. See \cite{collomb} for a first paper on $k$NN scalar regression, and \cite{kara_JMVA}, \cite{novo_2019} and \cite{linav20} for recent results on $k$NN nonparametric and semiparametric regression in the context of functional data.

The next three subsections present three $k$NN-based methodologies for estimating $\pmb{\beta}$ and $m(\cdot)$ in the SFPL-MAR model (\ref{eq_model1})-(\ref{MAR}), as well as asymptotic results and the assumptions considered to obtain them.
	
\subsection{Imputation estimators}\label{imputation}

\subsubsection{Construction of the estimators}\label{imputation-const}

The idea to develop imputation estimators is to use estimated regression values instead of the missing response values. For that, firstly one must to estimate the components ($\pmb{\beta}$ and $m(\cdot)$) in the regression function of model (\ref{eq_model1}) from the observed sample ($\{(Y_i,X_i,\mathcal{X}_i) \mbox{ such that } \delta_i=1\}$). We propose to estimate such components from the $k$NN-based estimators developed in \cite{linav20}, but adapted to the case of MAR responses. More specifically, the auxiliary estimators we propose in this first step are 
$$
\widehat{\pmb{\beta}}^C=\arg\underset{\pmb{\beta}}{\min}\sum_{i=1}^n\delta_i\left(Y_i -X_i^{\top}\pmb{\beta}-\breve{m}(\pmb{\beta},\mathcal{X}_i)\right)^2 \mbox{ and } \widehat{m}^C(\chi)=\breve{m}(\widehat{\pmb{\beta}}^C,\chi),
$$
where $
\breve{m}(\pmb{\beta},\chi)=\sum_{j=1}^n\delta_j\omega_{k_0}(\chi,\mathcal{X}_j)\left(Y_j-X_j^{\top}\pmb{\beta}\right).$ Note that such estimators can be written as  
\begin{equation}
	\widehat{\pmb{\beta}}^C=\left(\widetilde{\mathbf {X}}_0^{\top}\pmb{\delta}\widetilde{\mathbf {X}}_0 \right)^{-1}\widetilde{\mathbf {X}}_0^{\top}\pmb{\delta}\widetilde{\mathbf {Y}}_0
	\label{beta-C}
\end{equation}
and
\begin{equation}
	\widehat{m}^C(\chi)=\sum_{j=1}^n\delta_i\omega_{k_0}(\chi,\mathcal{X}_j)\left(Y_j-X_j^{\top}\widehat{\pmb{\beta}}^C\right), \label{m-C}
\end{equation}
where $\pmb{\delta}=\mbox{diag}(\delta_1,\ldots,\delta_n),$ $\mathbf{X}=(X_1,\ldots,X_n)^{\top}$ with $X_i=(X_{i1},\ldots,X_{ip})^{\top},$ $\mathbf{Y}=(Y_1,\ldots,Y_n)^{\top},$ 
\begin{equation}
   \omega_{k_0}(\chi,\mathcal{X}_j)=\frac{K_0\left(H_{k_0,\chi}^{-1}d\left(\mathcal{X}_j,\chi\right)\right)}{\sum_{i=1}^n\delta_iK_0\left(H_{k_0,\chi}^{-1}d\left(\mathcal{X}_i,\chi\right)\right)}
	\label{pesos-C}
\end{equation}
and, for any $(n\times q)$-matrix $\mathbf {A}$ $(q\geq 1)$, we have denoted
$\widetilde{\mathbf {A}}_u=\left(\mathbf{I}-\mathbf {W}_{u}\right)\mathbf {A}, \mbox{ where } \mathbf {W}_{u}=\left(\omega_{k_u}(\mathcal{X}_i,\mathcal{X}_j)\right)_{i,j}.$ 

The second step consists of imputing the estimated regression values (that is, considering $X_i^{\top}\widehat{\pmb{\beta}}^C+\widehat{m}^C(\mathcal{X}_i)$ instead of the missing responses $Y_i$ for $\delta_i=0$) and then construct the imputed estimators by means of the estimators proposed in \cite{linav20}, but considering
\begin{equation}
Y_i^I=\delta_iY_i+(1-\delta_i)\left(X_i^{\top}\widehat{\pmb{\beta}}^C+\widehat{m}^C(\mathcal{X}_i)  \right)
\end{equation}
as response in model (\ref{eq_model1}). That is, the proposed imputation estimators are
\begin{equation}
	\widehat{\pmb{\beta}}^I=\left(\widetilde{\mathbf {X}}_1^{\top}\widetilde{\mathbf {X}}_1 \right)^{-1}\widetilde{\mathbf {X}}_1^{\top}\widetilde{\mathbf {Y^I}}_1
	\label{beta-I}
\end{equation}
and
\begin{equation}
	\widehat{m}^I(\chi)=\sum_{j=1}^n\omega_{k_1}(\chi,\mathcal{X}_j)\left(Y_j^I-X_j^{\top}\widehat{\pmb{\beta}}^I\right),
	\label{m-I}
\end{equation}
where
\begin{equation}
	\omega_{k_1}(\chi,\mathcal{X}_j)=\frac{K_1\left(H_{k_1,\chi}^{-1}d\left(\mathcal{X}_j,\chi\right)\right)}{\sum_{i=1}^nK_1\left(H_{k_1,\chi}^{-1}d\left(\mathcal{X}_i,\chi\right)\right)}.
	\label{pesos-I}
\end{equation}
Note that, for the sake of generality, we allow the consideration of different kernel functions ($K_0(\cdot)$ and $K_1(\cdot)$) and also different numbers of neighbors ($k_0$ and $k_1$) in both the auxiliary and final estimators. Note also that our estimators (\ref{beta-I}) and (\ref{m-I}) extend the imputation estimators in \cite{wan07} to both the context of functional covariate $\mathcal{X}$ and $k$NN estimation, and also extend the estimators in \cite{lin19} from kernel to $k$NN estimation. Therefore, for the sake of brevity, we refer to \cite{wan07} and \cite{lin19} for a more precise justification of the construction of these estimators.

\subsubsection{Assumptions}\label{imputation-assum}
First of all, let us denote $g_j(\mathcal{X}_{i})=\mathbb{E}(X_{ij}|\mathcal{X}_{i})$, $\eta_{ij}=X_{ij}-g_j(\mathcal{X}_{i})$, $\eta_i=(\eta_{i1}, \ldots, \eta_{ip})^{\top}$, $g_j^C(\mathcal{X}_{i})=\mathbb{E}(\delta_i X_{ij}|\mathcal{X}_{i})/\mathbb{E}(\delta_i|\mathcal{X}_{i})$, $\eta_{ij}^C=X_{ij}-g_j^C(\mathcal{X}_{i})$, $\eta_i^C=(\eta_{i1}^C, \ldots, \eta_{ip}^C)^{\top}$ and $m^C(\mathcal{X}_{i})=\mathbb{E}(\delta_i \left(Y_i-X_i^{\top}\pmb{\beta}\right)|\mathcal{X}_{i})/\mathbb{E}(\delta_i|\mathcal{X}_{i})$ ($i=1,\ldots,n; j=1,\ldots,p$). In addition, let $\sigma^2(X,\mathcal{X})=\mathbb{E}(\varepsilon^2|X,\mathcal{X})$, $\Delta(x,\chi)=P(\delta=1|X=x,\mathcal{X}=\chi)$, $\Delta_1(\chi)=P(\delta=1|\mathcal{X}=\chi)$, $\pmb{\Sigma}_{0}=\mathbb{E}(\Delta(X_1,\mathcal{X}_{1})\eta_1^C\eta_1^{C {\top}})$, $\pmb{\Sigma}_{1}=\mathbb{E}(\eta_1\eta_1^{\top})$ and $\pmb{\Sigma}_{2}=\mathbb{E}((1-\Delta(X_1,\mathcal{X}_{1}))\eta_1\eta_1^{C {\top}}).$ Finally, for any set $S \subset \mathcal{F}$ and $\epsilon>0$, $\psi_{S}(\epsilon)$ denotes the Kolmogorov's $\epsilon$-entropy of $S$, which is defined as $\psi_{S}(\epsilon)=\log (N_{\epsilon}(S))$, where $N_{\varepsilon}(S)$ is the minimal number of open balls in $\mathcal{F}$ of radius $\epsilon$ which is necessary to cover $S$.

In order to prove our asymptotic results, we will use the following assumptions:

\noindent (A1) $\forall \epsilon >0,$ $\varphi_{\chi}(\epsilon):=P(\mathcal{X}\in B(\chi,\epsilon))>0$, with $\varphi_{\chi}(\cdot)$ continuous on a neighborhood of $0$ and $\varphi_{\chi}(0)=0.$

\noindent (A2) There exist a nonnegative function $\phi(\cdot)$ regularly varying at $0$ with nonnegative index, a positive function $g(\cdot)$ and a positive number $\alpha$ such that:
\begin{itemize}
	\item[(i)]$\phi (0)=0$ and $\lim_{\epsilon\rightarrow 0}\phi (\epsilon)=0.$
	\item[(ii)]$\exists ~C>0$ and $\exists \ \eta_{0}>0$ such that, $\forall 0<\eta<\eta_{0},$ $\phi^{'}(\eta)<C.$
	\item[(iii)]$\sup_{\chi\in S_{\mathcal{F}}}|\frac{\varphi_{\chi}(\epsilon)}{\phi (\epsilon)}-g(\chi)|=O(\epsilon^{\alpha}) \ \text{as} \ \epsilon \rightarrow 0$.
	\item[(iv)]$\exists C<\infty $ such that
	$\forall \text{ }\left( u,v\right) \in S_{\mathcal{F}}\times S_{\mathcal{F}},\ \forall \text{ }%
	f\in \left\{ m,g_{1},\ldots,g_{p}\right\} ,\text{ }\left\vert f\left( u\right)
	-f\left( v\right) \right\vert \leq Cd\left( u,v\right) ^{\alpha }.$
		\item[(v)]$\exists C<\infty $ such that
	$\forall \text{ }\left( u,v\right) \in S_{\mathcal{F}}\times S_{\mathcal{F}},\ \forall \text{ }%
	f\in \left\{m^C,g_{1}^C,\ldots,g_{p}^C\right\} ,\text{ }\left\vert f\left( u\right)
	-f\left( v\right) \right\vert \leq Cd\left( u,v\right) ^{\alpha }.$
\end{itemize}
(A3) The kernel function $K_u(\cdot)$ ($u=0,1$) satisfies:
\begin{itemize}
	\item[(i)]$K_u(\cdot)$ is a nonnegative, bounded and non increasing function with support $[0,1]$ and Lipschitz on $[0,1).$
	\item[(ii)]If $K_u(1)=0$, it must also be such that $-\infty<C<K_u'(t)<C'<0.$
\end{itemize}
(A4) The Kolmogorov's $\epsilon$-entropy of $S_{\mathcal{F}}$ satisfies:
\begin{equation}
	\sum_{n=1}^{n}\exp\{(1-\omega)\psi_{S_{\mathcal{F}}}(\frac{\log n}{n})\}<\infty \text{ for some } \omega >1. \nonumber
\end{equation}
(A5) $k_u=k_{u,n}$ ($u=0,1$) is a sequence of positive real numbers such that:
\begin{itemize}
	\item[(i)]$\frac{k_u}{n}\rightarrow 0$ and $ \frac{\log n}{k_u}\rightarrow 0$ as $n\rightarrow \infty$.
	\item[(ii)]For $n$ large enough, $\frac{(\log n)^{2}}{k_u}<\psi_{S_{\mathcal{F}}}(\frac{\log n}{n})<\frac{k_u}{\log n}.$
\end{itemize}
(A6) Moment conditions:
\begin{itemize}
	\item[(i)]
	$\forall r\geq 3, \ 1\leq j\leq p$ and $\chi \in S_{\mathcal{F}},$ 
	$\mathbb{E}(|Y_{1}|^{r}|\mathcal{X}_{1}=\chi)\leq C_{1}<\infty$ and $\mathbb{E}(|X_{1j}|^{r}|\mathcal{X}_{1}=\chi  )\leq C_{2}<\infty$.
	\item[(ii)]%
	$\pmb{\Sigma}_{0}$ and $\pmb{\Sigma}_{1}$ are positive definite matrices. 
\end{itemize}
(A7) The function $\Delta_1(\chi)$ satisfies $\inf_{\chi \in S_{\mathcal{F}}}\Delta_1(\chi)>0.$

\begin{remark}
Assumptions (A1)-(A6)(i) (except (A2)(v)) are usual ones in the context of $k$NN estimation in SFPL models (see, for instance, \citealt{linav20}), and the same applies to the condition on $\pmb{\Sigma}_{1}$ in assumption (A6)(ii). Assumptions (A2)(v) and (A7), as well as the condition on $\pmb{\Sigma}_{0}$ in assumption (A6)(ii), are usual in the specific context of missing data (see, for instance, \citealt{wan07}). 
\end{remark}

\subsubsection{Asymptotic results} \label{impu-asym}

\noindent The next theorem presents our theoretical results on the imputation estimator.
\begin{theorem} \label{theorem1} Under assumptions (A1)-(A7), if in addition, for $k\in \{k_0,k_1\}$, $\frac{\sqrt{n}\log^{2}n}{k}\rightarrow 0$, $\sqrt{n}\phi^{-1}(\frac{k}{n})^{\alpha}\rightarrow 0$ and $\frac{\sqrt{n}\psi_{S_{\mathcal{F}}}(\frac{\log n}{n})}{k}\rightarrow 0$ as $n\rightarrow \infty$, and $k\geq n^{(2/r)+b}/\log ^{2}n$ for $n$ large enough and some constant $b>0$ with $\frac{2}{r}+b>\frac{1}{2}$ (where $r\geq 3$), then we have:
	\begin{itemize}
		\item[(i)] $\sqrt{n}\left( \widehat{\pmb{\beta}}^I - \pmb{\beta}  \right)  \rightarrow N\left(\pmb{0},\pmb{\Sigma}_{1}^{-1} \pmb{V}_I   \pmb{\Sigma}_{1}^{-1}\right),$

\noindent where 
\begin{eqnarray}
\pmb{V}_I&=&\mathbb{E}(\Delta(X_1,\mathcal{X}_{1})\eta_1\eta_1^{\top}\sigma^2(X_1,\mathcal{X}_{1})) + \pmb{\Sigma}_{2}\pmb{\Sigma}_{0}^{-1}\mathbb{E}(\Delta(X_1,\mathcal{X}_{1})\eta_1^C\eta_1^{C {\top}}\sigma^2(X_1,\mathcal{X}_{1}))\pmb{\Sigma}_{0}^{-1}\pmb{\Sigma}_{2} \nonumber \\
&+& \mathbb{E}(\Delta(X_1,\mathcal{X}_{1})\eta_1\eta_1^{C {\top}}\sigma^2(X_1,\mathcal{X}_{1}))\pmb{\Sigma}_{0}^{-1}\pmb{\Sigma}_{2} + \pmb{\Sigma}_{2}\pmb{\Sigma}_{0}^{-1}\mathbb{E}(\Delta(X_1,\mathcal{X}_{1})\eta_1^C\eta_1^{\top}\sigma^2(X_1,\mathcal{X}_{1})). \nonumber
\end{eqnarray}
		
		\item[(ii)]$\sup_{\chi\in S_{\mathcal{F}}}|\widehat{m}^I(\chi)-m(\chi)|=O\left(\phi^{-1}\left(\frac{k_0}{n}\right) ^{\alpha}+\sqrt{\frac{\psi_{S_{\mathcal{F}}}(\frac{\log n}{n})}{k_0}}\right) + O\left(\phi^{-1}\left(\frac{k_1}{n}\right) ^{\alpha}+\sqrt{\frac{\psi_{S_{\mathcal{F}}}(\frac{\log n}{n})}{k_1}}\right) \ a.s.$

	\end{itemize}
\end{theorem}
Some comments on these results will be included in Section \ref{comm}.

\subsection{Semiparametric regression surrogate estimators}\label{surrogate}

\subsubsection{Construction of the estimators}\label{surrogate-const}

In the case of the regression surrogate estimators, the idea is to always use surrogate values for the response variable (even if the response variable is actually observed). Such surrogate values consist of estimated regression values.

Formally, let the surrogate response variable
\begin{equation}
Y_i^R=X_i^{\top}\widehat{\pmb{\beta}}^C+\widehat{m}^C(\mathcal{X}_i).
\end{equation}
Then, considering $Y_i^R$ as response variable in the SFPL model (\ref{eq_model1}),  from the procedure proposed in \cite{linav20} to estimate $\pmb{\beta}$ and $m(\cdot)$ when no missing data are present, one obtains the semiparametric regression surrogate estimators
\begin{equation}
	\widehat{\pmb{\beta}}^R=\left(\widetilde{\mathbf {X}}_1^{\top}\widetilde{\mathbf {X}}_1 \right)^{-1}\widetilde{\mathbf {X}}_1^{\top}\widetilde{\mathbf {Y^R}}_1
	\label{beta-R}
\end{equation}
and
\begin{equation}
	\widehat{m}^R(\chi)=\sum_{j=1}^n\omega_{k_1}(\chi,\mathcal{X}_j)\left(Y^R_j-X_j^{\top}\widehat{\pmb{\beta}}^R\right).\label{m-R}
\end{equation}
As in the case of the imputation estimators, our estimators (\ref{beta-R}) and (\ref{m-R}) extend the regression surrogate estimators in \cite{wan07} from the the setting where the variable $\mathcal{X}$ is scalar to the case where $\mathcal{X}$ is of functional nature. In addition, while the estimators in \cite{wan07} are based on kernel regression, the ones (\ref{beta-R}) and (\ref{m-R}) are based on $k$NN regression.

\subsubsection{Asymptotic results} \label{surro-asym}

In the following theorem, we state the asymptotic normality and the uniform almost sure rate of convergence of the proposed estimators (\ref{beta-R}) and (\ref{m-R}), respectively. Note that in the expression of $\pmb{V}_R,$ we have denoted $\pmb{\Sigma}_{3}=\mathbb{E}(\eta_1\eta_1^{C {\top}}).$
\begin{theorem} \label{theorem2} Under the assumptions of Theorem \ref{theorem1}, we have:
	
	\begin{itemize}
		\item[(i)] $\sqrt{n}\left( \widehat{\pmb{\beta}}^R - \pmb{\beta}  \right)  \rightarrow N\left(\pmb{0},\pmb{\Sigma}_{1}^{-1} \pmb{V}_R   \pmb{\Sigma}_{1}^{-1}\right),$

\noindent where $$\pmb{V}_R=\pmb{\Sigma}_{3} \pmb{\Sigma}_{0}^{-1}\mathbb{E}(\Delta(X_1,\mathcal{X}_{1})\eta_1^C\eta_1^{C {\top}}\sigma^2(X_1,\mathcal{X}_{1}))\pmb{\Sigma}_{0}^{-1} \pmb{\Sigma}_{3}.$$ %
		
		\item[(ii)]$\sup_{\chi\in S_{\mathcal{F}}}|\widehat{m}^R(\chi)-m(\chi)|=O\left(\phi^{-1}\left(\frac{k_0}{n}\right) ^{\alpha}+\sqrt{\frac{\psi_{S_{\mathcal{F}}}(\frac{\log n}{n})}{k_0}}\right) + O\left(\phi^{-1}\left(\frac{k_1}{n}\right) ^{\alpha}+\sqrt{\frac{\psi_{S_{\mathcal{F}}}(\frac{\log n}{n})}{k_1}}\right) \ a.s.$
		
	\end{itemize}
\end{theorem}

Some comments on these results will be included in Section \ref{comm}.

\subsection{Inverse marginal probability weighted estimators}\label{inverse}

\subsubsection{Construction of the estimators}
In the same way as in \cite{wan07}, we finally propose a third class of estimators for $\pmb{\beta}$ and $m(\cdot)$ in the SFPL-MAR model (\ref{eq_model1})-(\ref{MAR}). In the construction of such class, the inverse of an estimator of the marginal propensity score ($\Delta_1(\mathcal{X})$) plays a main role. 

More specifically, let
\begin{equation}
 Y_i^{IP}= \frac{\delta_i}{\widehat{\Delta}_1(\mathcal{X}_i)}Y_i+(1-\frac{\delta_i}{\widehat{\Delta}_1(\mathcal{X}_i)})\left(X_i^{\top}\widehat{\pmb{\beta}}^C+\widehat{m}^C(\mathcal{X}_i)  \right),
\end{equation}
where
\begin{equation}
\widehat{\Delta}_1(\mathcal{X}_i)=\sum_{j=1}^n\omega_{k_2}(\mathcal{X}_i,\mathcal{X}_j)\delta_j \label{est-Delta}
\end{equation}
with
\begin{equation}
	\omega_{k_2}(\chi,\mathcal{X}_j)=\frac{K_2\left(H_{k_2,\chi}^{-1}d\left(\mathcal{X}_j,\chi\right)\right)}{\sum_{i=1}^nK_2\left(H_{k_2,\chi}^{-1}d\left(\mathcal{X}_i,\chi\right)\right)}.
	\label{pesos-IP}
\end{equation}
The inverse marginal probability weighted estimators are defined in a similar way as the imputation estimators presented in Section \ref{imputation-const}, but with $Y_i^{I}$ replaced by $Y_i^{IP}.$ Therefore, the expressions of the proposed estimators for $\pmb{\beta}$ and $m(\cdot)$ are
\begin{equation}
	\widehat{\pmb{\beta}}^{IP}=\left(\widetilde{\mathbf {X}}_1^{\top}\widetilde{\mathbf {X}}_1 \right)^{-1}\widetilde{\mathbf {X}}_1^{\top}\widetilde{\mathbf {Y^{IP}}}_1
	\label{beta-IP}
\end{equation}
and
\begin{equation}
	\widehat{m}^{IP}(\chi)=\sum_{i=1}^n\omega_{k_1}(\chi,\mathcal{X}_i)\left(Y_i^{IP}-X_i^{\top}\widehat{\pmb{\beta}}^{IP}\right) \label{m-IP},
\end{equation}
respectively.

Note that a new kernel function ($K_2(\cdot)$) and a new smoothing parameter ($k_2$) are used in the nonparametric estimator of the marginal propensity score. 
Note also that the estimators (\ref{beta-IP}) and (\ref{m-IP}) extend in two ways the inverse marginal probability weight ones presented in \cite{wan07}: functional variable $\mathcal{X}$ instead of scalar one, and $k$NN regression instead of kernel one.
\subsubsection{Additional assumptions}\label{inverse-assum}
In addition to the assumptions used in previous theorems (see Section \ref{imputation-assum}), we will need the following new assumptions to guarantee appropriate asymptotic behavior of the estimator (\ref{est-Delta}), which takes part in the proposed estimators (\ref{beta-IP}) and (\ref{m-IP}).

\noindent (A8) The kernel function $K_2(\cdot)$ satisfies:
\begin{itemize}
	\item[(i)]$K_2(\cdot)$ is a nonnegative, bounded and non increasing function with support $[0,1]$ and Lipschitz on $[0,1).$
	\item[(ii)]If $K_2(1)=0$, it must also be such that $-\infty<C<K_2'(t)<C'<0.$
\end{itemize}
(A9) $k_2=k_{2,n}$ is a sequence of positive real numbers such that:
\begin{itemize}
	\item[(i)]$\frac{k_2}{n}\rightarrow 0$ and $ \frac{\log n}{k_2}\rightarrow 0$ as $n\rightarrow \infty$.
	\item[(ii)]For $n$ large enough, $\frac{(\log n)^{2}}{k_2}<\psi_{S_{\mathcal{F}}}(\frac{\log n}{n})<\frac{k_2}{\log n}.$
\end{itemize}
\noindent(A10) $\exists C<\infty $ such that
	$\forall \text{ }\left( u,v\right) \in S_{\mathcal{F}}\times S_{\mathcal{F}},\ \left\vert \Delta_1\left( u\right)-\Delta_1\left( v\right) \right\vert \leq Cd\left( u,v\right) ^{\alpha }.$

\subsubsection{Asymptotic results} \label{inver-asym}
The following theorem presents our asymptotic results related to estimators (\ref{beta-IP}) and (\ref{m-IP}). Note that in the expression of $\pmb{V}_{IP}$, we have denoted $\pmb{\Sigma}_4=\mathbb{E}((1-\frac{\delta_1}{\Delta_1(\mathcal{X}_1)}))\eta_1\eta_1^{C {\top}}).$

\begin{theorem} \label{theorem3} Under assumptions (A1)-(A10), if in addition, for $k\in \{k_0,k_1,k_2\}$, $\frac{\sqrt{n}\log^{2}n}{k}\rightarrow 0$, $\sqrt{n}\phi^{-1}(\frac{k}{n})^{\alpha}\rightarrow 0$ and $\frac{\sqrt{n}\psi_{S_{\mathcal{F}}}(\frac{\log n}{n})}{k}\rightarrow 0$ as $n\rightarrow \infty$, and $k\geq n^{(2/r)+b}/\log ^{2}n$ for $n$ large enough and some constant $b>0$ with $\frac{2}{r}+b>\frac{1}{2}$ (where $r\geq 3$), then we have:
	\begin{itemize}
		\item[(i)] $\sqrt{n}\left( \widehat{\pmb{\beta}}^{IP} - \pmb{\beta}  \right)  \rightarrow N\left(\pmb{0},\pmb{\Sigma}_{1}^{-1} \pmb{V}_{IP}   \pmb{\Sigma}_{1}^{-1}\right),$

\noindent where
\begin{eqnarray}
\pmb{V}_{IP}&=&\mathbb{E}\left(\frac{\Delta(X_1,\mathcal{X}_{1})}{\Delta_1^2(\mathcal{X}_{1})}\eta_1\eta_1^{\top}\sigma^2(X_1,\mathcal{X}_{1})\right) + \pmb{\Sigma}_{4}\pmb{\Sigma}_{0}^{-1}\mathbb{E}(\Delta(X_1,\mathcal{X}_{1})\eta_1^C\eta_1^{C {\top}}\sigma^2(X_1,\mathcal{X}_{1}))\pmb{\Sigma}_{0}^{-1}\pmb{\Sigma}_{4} \nonumber \\
&+& \mathbb{E}\left(\frac{\Delta(X_1,\mathcal{X}_{1})}{\Delta_1(\mathcal{X}_{1})}\eta_1\eta_1^{C {\top}}\sigma^2(X_1,\mathcal{X}_{1})\right)\pmb{\Sigma}_{0}^{-1}\pmb{\Sigma}_{4} + \pmb{\Sigma}_{4}\pmb{\Sigma}_{0}^{-1}\mathbb{E}\left(\frac{\Delta(X_1,\mathcal{X}_{1})}{\Delta_1(\mathcal{X}_{1})}\eta_1^C\eta_1^{\top}\sigma^2(X_1,\mathcal{X}_{1})\right). \nonumber
\end{eqnarray}
		
		\item[(ii)]$\sup_{\chi\in S_{\mathcal{F}}}|\widehat{m}^{IP}(\chi)-m(\chi)|=O\left(\phi^{-1}\left(\frac{k_0}{n}\right) ^{\alpha}+\sqrt{\frac{\psi_{S_{\mathcal{F}}}(\frac{\log n}{n})}{k_0}}\right) + O\left(\phi^{-1}\left(\frac{k_1}{n}\right) ^{\alpha}+\sqrt{\frac{\psi_{S_{\mathcal{F}}}(\frac{\log n}{n})}{k_1}}\right) \ a.s.$
		
	\end{itemize}
\end{theorem}
Some comments on these results will be included in the next Section \ref{comm}.

\subsection{Some comments on the obtained results}\label{comm}

The asymptotic study presented in sections \ref{impu-asym}, \ref{surro-asym} and \ref{inver-asym} shows some similarities and some differences between the proposed estimators. On the one hand, the three estimators of $\pmb{\beta}$ ((\ref{beta-I}), (\ref{beta-R}) and (\ref{beta-IP})) are $\sqrt{n}$-consistent. This fact indicates that the presence of the nonparametric component $m(\cdot)$ in the SFPL-MAR model (\ref{eq_model1})-(\ref{MAR}) does not affect the rate of convergence of the estimators of $\pmb{\beta}$. In the same way, such estimators maintain the asymptotic normality, the difference between the corresponding distributions being in the asymptotic variances. On the other hand, the three estimators of $m(\cdot)$ ((\ref{m-I}), (\ref{m-R}) and (\ref{m-IP})) converge at the same uniform rate. If one considers the same values for the several smoothing parameters ($k_1=k_2=k_3=k$), the corresponding rate of convergence is the same as that obtained in \cite{Kud13} and \cite{linav20} for the pure functional nonparametric model and for the SFPL model (\ref{eq_model1}), respectively (in these two papers, data were completely observed). It is worth being noted that such uniform convergence rate is almost sure, while in \cite{wan07} and \cite{lin19} the uniform rate was in probability.

Now we focus on the asymptotic variance of the three estimators of $\pmb{\beta}$. Given the complexity of such variances, we will discuss some special cases. In the case where $\delta_i$ is independent of $X_i$ given $\mathcal{X}_i$, one has that $\eta_1=\eta_1^C$ and $\Delta(X_1,\mathcal{X}_1)=\Delta_1(\mathcal{X}_1)$. Then, it is easy to obtain that the asymptotic variances of both the imputation ($\widehat{\pmb{\beta}}^I$) and the surrogate ($\widehat{\pmb{\beta}}^R$) estimators reduce to $\pmb{\Sigma}_{01}^{-1}\mathbb{E}(\Delta_1(\mathcal{X}_{1})\eta_1\eta_1^{\top}\sigma^2(X_1,\mathcal{X}_{1}))\pmb{\Sigma}_{01}^{-1},$ where $\pmb{\Sigma}_{01}=\mathbb{E}(\Delta_1(\mathcal{X}_{1})\eta_1\eta_1^{\top}),$ while the asymptotic variance of the inverse marginal probability weighted estimator ($\widehat{\pmb{\beta}}^{IP}$) reduces to $\pmb{\Sigma}_{1}^{-1}\mathbb{E}(\eta_1\eta_1^{\top}\frac{\sigma^2(X_1,\mathcal{X}_{1})}{\Delta_1(\mathcal{X}_{1})})\pmb{\Sigma}_{1}^{-1}.$ If, in addition, $\Delta_1(\mathcal{X}_1)$ is equal to a constant $a$, then the three estimators have the same asymptotic variance, $\frac{1}{a}\pmb{\Sigma}_{1}^{-1}\mathbb{E}(\eta_1\eta_1^{\top}\sigma^2(X_1,\mathcal{X}_{1}))\pmb{\Sigma}_{1}^{-1}.$ In the case $a=1$, such common variance coincides with the variance of the estimator when the data are completely observed (see \citealt{linav20}). Note that these conclusions coincide with those in \cite{wan07}.

\bigskip
\noindent{\large \bf{Acknowledgements}}
\noindent This research/work is part of the grant PID2023-147127OB-I00 ``ERDF/EU'', funded by MCIN/AEI/10.13039/501100011033/. It has also been supported by the Xunta de Galicia (Grupos de Referencia Competitiva ED431C-2024/14) and by CITIC as a center accredited for excellence within the Galician University System and a member of the CIGUS Network, receives subsidies from the Department of Education, Science, Universities, and Vocational Training of the Xunta de Galicia. Additionally, it is co-financed by the EU through the FEDER Galicia 2021-27 operational program (Ref. ED431G 2023/01).

\appendix

\section{Appendix: Proofs}\label{app}

In this Appendix, we first present some technical lemmas to be used in the proofs of our theorems. Some of such lemmas are known while the others are slight modifications of known lemmas/theorems. Then, we obtain the proofs of our main results.

\subsection{Technical lemmas}\label{app1}

\begin{lemma}
	\label{LemaStout} (\citealt{stout74}, Corollary 5.2.3) Let $V_{1},...,V_{n}$ be independent r.v. with 0 means and
	$\max_{1\leq i\leq n}E\left\vert V_{i}\right\vert ^{2+\delta}<\infty$, for
	some $\delta>0$. If, in addition, $\lim\inf_{n\rightarrow\infty}n^{-1}\sum_{i=1}^{n}Var\left(
	V_{i}\right)  >0,$ then%
	\[
	\lim\sup_{n\rightarrow\infty}\left\vert S_{n}\right\vert /\left(  2s_{n}%
	^{2}\log\log s_{n}^{2}\right)  ^{1/2}=1\text{ a.s.,}%
	\]
	where $S_{n}=\sum_{i=1}^{n}V_{i}$ and $s_{n}^{2}=\sum_{i=1}^{n}Var\left(
	V_{i}\right)  $.
\end{lemma}

\begin{lemma}\label{lemma4} (\citealt{anev06}, Lemma 3) Let $V_{1},\ldots,V_{n}$ be independent r.v. with zero means and such
	that for some $r\geq2$, $\max_{1\leq j\leq n}E\left\vert V_{j}\right\vert
	^{r}\leq C<\infty$. Assume that $\left\{  a_{ij},\text{ }i,j=1,\ldots,n\right\}
	$ is a sequence of positive numbers such that $\max_{1\leq i,j\leq
		n}\left\vert a_{ij}\right\vert =O\left(  a_{n}\right)  $ and $\max_{1\leq
		i\leq n}\sum_{j=1}^{n}\left\vert a_{ij}\right\vert =O\left(  b_{n}\right)  $.
	If, in addition,
	\begin{equation}
		\exp\left(  -\dfrac{b_{n}^{1/2}\left(  \log n\right)  ^{2}%
		}{b_{n}^{1/2}+a_{n}^{1/2}n^{1/r}\log n}\right)  =O\left(  n^{-a}\right) ,
		(a>2) ,\nonumber\\
	\end{equation}
	and
	\begin{equation}
		a_{n}^{1/2}n^{1/r+b}=O\left(  b_{n}^{1/2}\log n\right),
		(b>0),\nonumber\\
	\end{equation}
	then
	\begin{equation}
		\max_{1\leq i\leq n}\left\vert \sum_{j=1}^{n}a_{ij}V_{j}\right\vert =O\left(
		a_{n}^{1/2}b_{n}^{1/2}\log n\right)  \text{ a.s.}\nonumber\\
	\end{equation}
	As a matter of fact, the conclusion of Lemma \ref{lemma4} remains unchanged
	when $\left\{  a_{ij},\text{ }i,j=1,\ldots,n\right\}  $ is a random sequence
	satisfying the conditions above almost surely.
\end{lemma}

\begin{lemma}\label{lemma1} (\citealt{Kud13}, Theorem 2) If assumptions (A1)-(A6)(i) (except (A2)(v) and also except both (A3) and (A5) for $u=0$) are satisfied then we have that\\
	\begin{equation}
		\sup_{\chi\in S_{\mathcal{F}}}\left|g_{j}(\chi)-\sum_{i=1}^{n}\omega_{k_1}(\chi,\mathcal{X}_i)X_{ij} \right|= O\left(\phi^{-1}\left(\frac{k_1}{n}\right) ^{\alpha}+\sqrt{\frac{\psi_{S_{\mathcal{F}}}(\frac{\log n}{n})}{k_1}}\right) \ a.s. \ (j=0,1,\ldots,p), \nonumber
	\end{equation}
	where $g_0(\cdot)=m(\cdot)$ and $X_{i0}=g_0(\mathcal{X}_i)+\varepsilon_i.$ 
\end{lemma}

\begin{lemma}\label{lemma1.5} If assumptions (A1)-(A6)(i) (except (A2)(iv) and also except both (A3) and (A5) for $u=1$) and (A7) are satisfied then we have that
	\begin{equation}
		\sup_{\chi\in S_{\mathcal{F}}}\left|g_{j}(\chi)-\sum_{i=1}^{n}\delta_i\omega_{k_0}(\chi,\mathcal{X}_i)X_{ij} \right|= O\left(\phi^{-1}\left(\frac{k_0}{n}\right) ^{\alpha}+\sqrt{\frac{\psi_{S_{\mathcal{F}}}(\frac{\log n}{n})}{k_0}}\right) \ a.s. \ (j=0,1,\ldots,p), \nonumber
	\end{equation}
where $g_0(\cdot)=m(\cdot)$ and $X_{i0}=g_0(\mathcal{X}_i)+\varepsilon_i.$ 
\end{lemma}
\noindent \emph{Proof.} Lemma \ref{lemma1.5} can be seen as a special case of Theorem 2 in \cite{rac21}, where almost complete rates of uniform convergence for the nonparametric local linear estimator are obtained. Indeed, the later can be obtained by taking $b=0$ in \cite{rac21}. $\Box$

\begin{lemma}\label{lemma2} (\citealt{linav20}, Lemma 2) Under Assumption (A3) with $u=1$ we have that\\
	\begin{equation}
		\max_{1\leq i,j\leq n}| \omega_{k_1}(\mathcal{X}_i,\mathcal{X}_j) |=O\left(\frac{1}{k_1}\right). \nonumber
	\end{equation}
\end{lemma}

\begin{lemma}\label{lemma3} (\citealt{linav20}, Lemma 3) Under conditions (A1)-(A5) (except (A2)(v) and also except both (A3) and (A5) for $u=0$), if, in addition, $\forall r\geq 3, \ 1\leq j\leq p,$ $\mathbb{E}(|X_{1j}|^{r}|\mathcal{X}_{1}=\chi  )\leq C<\infty,$ then we have that\\
	\begin{equation}
		\frac{1}{n}\widetilde{\mathbf{X}}_1^{T}\widetilde{\mathbf{X}}_1\rightarrow \pmb{\Sigma}_{1} \text{ a.s}. \nonumber 
	\end{equation}
\end{lemma}

\begin{lemma}\label{lemma3.5} Under conditions (A1)-(A5) (except (A2)(iv) and also except both (A3) and (A5) for $u=1$), if, in addition, $\forall r\geq 3, \ 1\leq j\leq p,$ $\mathbb{E}(|\delta_iX_{1j}|^{r}|\mathcal{X}_{1}=\chi  )\leq C<\infty,$ then we have that\\
	\begin{equation}
		\frac{1}{n}\widetilde{\mathbf{X}}_0^{T}\pmb{\delta}\widetilde{\mathbf{X}}_0\rightarrow \pmb{\Sigma}_{0} \text{ a.s}. \nonumber
	\end{equation}
\end{lemma}
\noindent \emph{Proof.} Taking into account that $\eta_{ij}^C$ and $g_j^C$ ($i=1,\ldots,n$ and $j=1,\ldots,p$) satisfy the same assumptions as $\eta_{ij}$ and $g_j$, respectively, and using Lemma \ref{lemma1.5} instead of Lemma \ref{lemma1} together with the fact that, by the strong law of large numbers, $$n^{-1}\sum_{t=1}^n \eta_{tu}^C\eta_{tv}^C \rightarrow (\pmb{\Sigma}_{0})_{uv} \text{ a.s},$$ where $(\pmb{\Sigma}_{0})_{uv}=\mathbb{E}(\Delta(X_1,\mathcal{X}_{1})\eta_{1u}^C\eta_{1v}^{C})$ is the $(u, v)$-th element of $\pmb{\Sigma}_{0}$, the proof of this lemma is similar to that of Lemma \ref{lemma3}. Therefore we omit it here. $\Box$

\begin{lemma}\label{lemma3.7} Under assumptions (A1)-(A7) (except (A2)(iv) and also except both (A3) and (A5) for $u=1$), if in addition $\frac{\sqrt{n}\log^{2}n}{k_0}\rightarrow 0$, $\sqrt{n}\phi^{-1}(\frac{k_0}{n})^{\alpha}\rightarrow 0$ and $\frac{\sqrt{n}\psi_{S_{\mathcal{F}}}(\frac{\log n}{n})}{k_0}\rightarrow 0$ as $n\rightarrow \infty$, and $k_0\geq n^{(2/r)+b}/\log ^{2}n$ for $n$ large enough and some constant $b>0$ with $\frac{2}{r}+b>\frac{1}{2}$ (where $r\geq 3$), then we have:
	\begin{equation}
		\sqrt{n}\left( \widehat{\pmb{\beta}}^C - \pmb{\beta}  \right)  \rightarrow N\left(\pmb{0},\pmb{\Sigma}_{0}^{-1} \pmb{V}_C   \pmb{\Sigma}_{0}^{-1}\right) \label{betaC-norm}
	\end{equation}
	and
	\begin{equation}
	\limsup_{n\rightarrow \infty}\big(\frac{n}{2\log \log n}\big)^{\frac{1}{2}}|\widehat{\beta}^C_j - {\beta}_j|=(\sigma_{jj})^{\frac{1}{2}} \ a.s. \ (j=1,\ldots,p), \label{betaC-log}
	\end{equation}
where $$\pmb{V}_C=\mathbb{E}(\Delta(X_1,\mathcal{X}_{1})\eta_1^C\eta_1^{C {\top}}\sigma^2(X_1,\mathcal{X}_{1}))$$ and
		$\beta_{j}$ and $\widehat{\beta}^C_j$ denote the $j$-th components of the vectors $\pmb{\beta}$ and $\widehat{\pmb{\beta}}^C$, respectively, while $\sigma_{jj}$ denotes the $(j,j)$-element of the matrix $\pmb{\Sigma}_{0}^{-1} \pmb{V}_C \pmb{\Sigma}_{0}^{-1}.$
\end{lemma}
\noindent \emph{Proof.} Taking into account that $\eta_{ij}^C$ and $g_j^C$ ($i=1,\ldots,n$; $j=1,\ldots,p$) satisfy the same assumptions as $\eta_{ij}$ and $g_j$, respectively, and using Lemma \ref{lemma1.5} and Lemma \ref{lemma3.5} instead of Lemma \ref{lemma1} and Lemma \ref{lemma3}, respectively, the proof of this lemma is similar to that of Theorem 1(i, ii) in \cite{linav20}. Therefore we omit it here. $\Box$

\subsection{Proof of Theorem \ref{theorem1}}\label{app2}
Before presenting the proof of Theorem \ref{theorem1}, we introduce some notation that will be used in this and other proofs. Let us denote
$$\widehat{f}^C(\chi)=\sum_{i=1}^{n}\delta_i\omega_{k_0}(\chi,\mathcal{X}_i)Y_i, \	
	\widehat{g}^C_j(\chi)=\sum_{i=1}^{n}\delta_i\omega_{k_0}(\chi,\mathcal{X}_i)X_{ij}, \ \widehat{\pmb{g}}^C(\chi)=(\widehat{g}^C_1(\chi),\ldots,\widehat{g}^C_p(\chi))^{\top},$$
$$f(\chi)=\mathbb{E}(Y |\mathcal{X}=\chi ), \	
	\widehat{g}_j(\chi)=\sum_{i=1}^{n}\omega_{k_1}(\chi,\mathcal{X}_i)X_{ij}, \ \widehat{\pmb{g}}(\chi)=(\widehat{g}_1(\chi),\ldots,\widehat{g}_p(\chi))^{\top}, \ \pmb{g}(\chi)=(g_1(\chi),\ldots,g_p(\chi))^{\top},$$
$$\widetilde{\mathbf{X}}_1=(\widetilde{X}_1,\ldots,\widetilde{X}_n)^{\top}, \ \pmb{\eta}=(\eta_1,\ldots,\eta_n)^{\top}, \ \widetilde{\mathbf{G}}_1=(\widetilde{G}_{ij}), \mbox{ where } \widetilde{G}_{ij}=g_j(\mathcal{X}_i)-\sum_{s=1}^n\omega_{k_1}(\mathcal{X}_i,\mathcal{X}_s)g_j(\mathcal{X}_s),$$
\begin{equation}
 \ Y_i^{I\ast}=\delta_iY_i+(1-\delta_i)\left(X_i^{\top}\pmb{\beta}+m(\mathcal{X}_i)  \right) \mbox{ and } \widehat{m}^{C\ast}(\mathcal{X}_i) =\widehat{f}^C(\mathcal{X}_i)-\widehat{\pmb{g}}^C(\mathcal{X}_i)^{\top}\pmb{\beta},
\end{equation}

\noindent \textit{Proof of Theorem \ref{theorem1}(i).} We have that

\begin{equation}
n^{1/2}(\widehat{\pmb{\beta}}^{I}-\pmb{\beta})=\left(n^{-1}\widetilde{\mathbf{X}}_1^{\top}\widetilde{\mathbf{X}}_1\right)^{-1} n^{-1/2} A_n, \label{betaI-beta}
\end{equation}
where 
$$A_n=\widetilde{\mathbf{X}}_1^{\top}\left(\widetilde{\mathbf{Y^I}}_1-\widetilde{\mathbf{X}}_1\pmb{\beta}\right).$$
Then we can write 
\begin{equation}
A_n=\sum_{r=1}^7A_{nr}, \label{A}
\end{equation}
where
$$A_{n1}=\sum_{i=1}^n\widetilde{X}_i(1-\delta_i)(\widehat{m}^{C\ast}(\mathcal{X}_i)-m(\mathcal{X}_i)),
$$

$$A_{n2}=-\sum_{i=1}^n\widetilde{X}_i\left(\sum_{j=1}^n\omega_{k_1}(\mathcal{X}_i,\mathcal{X}_j) Y_j^{I\ast}-f(\mathcal{X}_i)   \right),
$$

$$A_{n3}=-\sum_{i=1}^n\widetilde{X}_i\left(\sum_{j=1}^n\omega_{k_1}(\mathcal{X}_i,\mathcal{X}_j) (1-\delta_j)X_j^{\top}(\widehat{\pmb{\beta}}^C-\pmb{\beta})\right), 
$$

$$A_{n4}=-\sum_{i=1}^n\widetilde{X}_i\left(\sum_{j=1}^n\omega_{k_1}(\mathcal{X}_i,\mathcal{X}_j) (1-\delta_j)(\widehat{m}^C(\mathcal{X}_j)-m(\mathcal{X}_j))\right),
$$

$$A_{n5}=\sum_{i=1}^n\widetilde{X}_i\left(\widehat{\pmb{g}}(\mathcal{X}_i)-\pmb{g}(\mathcal{X}_i)\right)^{\top}\pmb{\beta},
$$

$$A_{n6}=\sum_{i=1}^n\widetilde{X}_i(1-\delta_i)(X_i-\widehat{\pmb{g}}^C(\mathcal{X}_i))^{\top}(\widehat{\pmb{\beta}}^C-\pmb{\beta}) 
$$
and
$$A_{n7}=\sum_{i=1}^n\widetilde{X}_i\delta_i\varepsilon_i.
$$
In addition, taking into account that
\begin{equation}
\widetilde{\mathbf{X}}_1=\widetilde{\mathbf{G}}_1+\pmb{\eta}-\mathbf{W}_1\pmb{\eta}, \label{Xtilde-3}
\end{equation}
we have that each component $A_{nr}$ ($r=1,\ldots,7$) can be decomposed into three summands $A_{nr}^{(u)}$ ($u=1,2,3$). If we add the subscript $j$ ($1 \leq j \leq p$) to indicate the $j$-the component of each term $A_{nr}$ and $A_{nr}^{(u)}$, we have that   
\begin{equation}
A_{nrj}=\sum_{i=1}^n\widetilde{G}_{ij}c_{ri} + \sum_{i=1}^n\eta_{ij}c_{ri} + \sum_{i=1}^n (\sum_{s=1}^n \omega_{k_1}(\mathcal{X}_i,\mathcal{X}_s)\eta_{sj})c_{ri}=A_{nrj}^{(1)}+A_{nrj}^{(2)}+A_{nrj}^{(3)}.\label{A1-7}
\end{equation}

Firstly, we present some results that will play a main role in the proofs of this theorem. From Lemma \ref{lemma1} we have that
\begin{equation}
\max_{1\leq i \leq n}|\widetilde{G}_{ij}|= O\left(\phi^{-1}\left(\frac{k_1}{n}\right) ^{\alpha}+\sqrt{\frac{\psi_{S_{\mathcal{F}}}(\frac{\log n}{n})}{k_1}}\right) \ a.s. \ (j=1,\ldots,p), \label{O_gtilde}
\end{equation}
\begin{equation}
\sup_{\chi\in S_{\mathcal{F}}}|\widehat{g}_{j}(\chi)-g_j(\chi)|= O\left(\phi^{-1}\left(\frac{k_1}{n}\right) ^{\alpha}+\sqrt{\frac{\psi_{S_{\mathcal{F}}}(\frac{\log n}{n})}{k_1}}\right) \ a.s. \ (j=1,\ldots,p) \label{O_ghat}
\end{equation}
and
\begin{equation}
\sup_{\chi\in S_{\mathcal{F}}}|\sum_{j=1}^n\omega_{k_1}(\chi,\mathcal{X}_j) Y_j^{I\ast}-f(\chi) |= O\left(\phi^{-1}\left(\frac{k_1}{n}\right) ^{\alpha}+\sqrt{\frac{\psi_{S_{\mathcal{F}}}(\frac{\log n}{n})}{k_1}}\right) \ a.s. \  \label{O_ftildeast}
\end{equation}
hold (note that in (\ref{O_ftildeast}) we have used the fact that, under the MAR condition, $\mathbb{E}\left(Y_i^{I\ast}|\mathcal{X}_i\right)=f(\mathcal{X}_i)$ holds). In addition, taking into account that $\widehat{m}^{C\ast}(\chi)=\sum_{j=1}^{n}\delta_j\omega_{k_0}(\chi,\mathcal{X}_j)(m(\mathcal{X}_j)+\varepsilon_j)$, Lemma \ref{lemma1.5} gives
\begin{equation}
\sup_{\chi\in S_{\mathcal{F}}}
|\widehat{m}^{C\ast}(\chi)-m(\chi)|=O\left(\phi^{-1}\left(\frac{k_0}{n}\right) ^{\alpha}+\sqrt{\frac{\psi_{S_{\mathcal{F}}}(\frac{\log n}{n})}{k_0}}\right) \ a.s. \label{O_mtildeCast}
\end{equation}
Now, from (\ref{O_mtildeCast}), Lemma \ref{lemma1.5} and the result (\ref{betaC-log}) in Lemma \ref{lemma3.7}, we can establish that
\begin{eqnarray}
\sup_{\chi\in S_{\mathcal{F}}}|\widehat{m}^C(\chi)-m(\chi)| &\leq& \sup_{\chi\in S_{\mathcal{F}}}|\widehat{m}^{C\ast}(\chi)-m(\chi)| + \sup_{\chi\in S_{\mathcal{F}}}|\widehat{\pmb{g}}^C(\chi)^{\top}(\widehat{\pmb{\beta}}^C-\pmb{\beta})| \nonumber\\
&=&O\left(\phi^{-1}\left(\frac{k_0}{n}\right) ^{\alpha}+\sqrt{\frac{\psi_{S_{\mathcal{F}}}(\frac{\log n}{n})}{k_0}}\right) \ a.s. \label{O_mtilde}
\end{eqnarray}
Finally, considering $a_{is}=\omega_{k_1}(\mathcal{X}_i,\mathcal{X}_s)$, $a_n=1/k_1$ (see Lemma \ref{lemma2}) and $b_n=1$ in Lemma \ref{lemma4}, we obtain that
\begin{equation}
\max_{1\leq i \leq n}\sum_{s=1}^n\omega_{k_1}(\mathcal{X}_i,\mathcal{X}_s)\eta_{s,j}=O\left(\frac{\log n}{\sqrt{k_1}}\right) \ a.s. \label{O-eta}
\end{equation}

Next we focus on the terms $A_{nr}$ ($r=1,2,4,5$), which can be studied in a similar way. On the one hand, from (\ref{O_mtildeCast}), (\ref{O_ftildeast}), (\ref{O_mtilde}) and (\ref{O_ghat}) we have that
\begin{equation}
\max_{1\leq i \leq n}|c_{1i}|=O\left(\phi^{-1}\left(\frac{k_0}{n}\right) ^{\alpha}+\sqrt{\frac{\psi_{S_{\mathcal{F}}}(\frac{\log n}{n})}{k_0}}\right)\ a.s, \label{c1}
\end{equation}
\begin{equation}
\max_{1\leq i \leq n}|c_{2i}|=O\left(\phi^{-1}\left(\frac{k_1}{n}\right) ^{\alpha}+\sqrt{\frac{\psi_{S_{\mathcal{F}}}(\frac{\log n}{n})}{k_1}}\right) \ a.s, \label{c2}
\end{equation}
\begin{equation}
\max_{1\leq i \leq n}|c_{4i}|=O\left(\phi^{-1}\left(\frac{k_0}{n}\right) ^{\alpha}+\sqrt{\frac{\psi_{S_{\mathcal{F}}}(\frac{\log n}{n})}{k_0}}\right) \ a.s \label{c4}
\end{equation}
and
\begin{equation}
\max_{1\leq i \leq n}|c_{5i}|=O\left(\phi^{-1}\left(\frac{k_1}{n}\right) ^{\alpha}+\sqrt{\frac{\psi_{S_{\mathcal{F}}}(\frac{\log n}{n})}{k_1}}\right) \ a.s, \label{c5}
\end{equation}
respectively. On the other hand, combining Abel's inequality with both (\ref{O_gtilde}) and (\ref{c1}), we obtain that
\begin{center}
	\begin{eqnarray}
\left |A_{n1j}^{(1)}\right | &\leq& n \max_{1 \leq i \leq n}|\widetilde{G}_{ij}|\max_{1 \leq i \leq n}|c_{1i}|\nonumber
\\
&=&O\left(n\left(\phi^{-1}\left(\frac{k_1}{n}\right) ^{\alpha}+\sqrt{\frac{\psi_{S_{\mathcal{F}}}(\frac{\log n}{n})}{k_1}}\right)\left(\phi^{-1}\left(\frac{k_0}{n}\right) ^{\alpha}+\sqrt{\frac{\psi_{S_{\mathcal{F}}}(\frac{\log n}{n})}{k_0}}\right)\right)=o(\sqrt{n}) \ a.s, \label{A11}
\end{eqnarray}
\end{center}
In the same way, combining Abel's inequality with both (\ref{O_gtilde}) and each of the three results (\ref{c2})-(\ref{c5}), we obtain that
\begin{equation}
A_{n2j}^{(1)}=O\left(n\left(\phi^{-1}\left(\frac{k_1}{n}\right) ^{\alpha}+\sqrt{\frac{\psi_{S_{\mathcal{F}}}(\frac{\log n}{n})}{k_1}}\right)^2\right)=o(\sqrt{n}) \ a.s, \label{A21}
\end{equation}
\begin{equation}
A_{n4j}^{(1)}=O\left(n\left(\phi^{-1}\left(\frac{k_1}{n}\right) ^{\alpha}+\sqrt{\frac{\psi_{S_{\mathcal{F}}}(\frac{\log n}{n})}{k_1}}\right)\left(\phi^{-1}\left(\frac{k_0}{n}\right) ^{\alpha}+\sqrt{\frac{\psi_{S_{\mathcal{F}}}(\frac{\log n}{n})}{k_0}}\right)\right)=o(\sqrt{n}) \ a.s \label{A41}
\end{equation}
and
\begin{equation}
A_{n5j}^{(1)}=O\left(n\left(\phi^{-1}\left(\frac{k_1}{n}\right) ^{\alpha}+\sqrt{\frac{\psi_{S_{\mathcal{F}}}(\frac{\log n}{n})}{k_1}}\right)^2\right)=o(\sqrt{n}) \ a.s. \label{A51}
\end{equation}
respectively. In addition, considering $a_{is}=c_{1i}$ and $b_n=na_n$ in Lemma \ref{lemma4}, and taking (\ref{c1}) into account, we obtain that
\begin{equation}
A_{n1j}^{(2)}=O\left(n^{1/2}\log n\left(\phi^{-1}\left(\frac{k_0}{n}\right) ^{\alpha}+\sqrt{\frac{\psi_{S_{\mathcal{F}}}(\frac{\log n}{n})}{k_0}}\right)\right) = o(\sqrt{n}) \ a.s.  \label{A12}
\end{equation}
In the same way, taking (\ref{c2})-(\ref{c5}) into account, we have that
\begin{equation}
A_{n2j}^{(2)}=O\left(n^{1/2}\log n\left(\phi^{-1}\left(\frac{k_1}{n}\right) ^{\alpha}+\sqrt{\frac{\psi_{S_{\mathcal{F}}}(\frac{\log n}{n})}{k_1}}\right)\right) = o(\sqrt{n}) \ a.s.,  \label{A22}
\end{equation}
\begin{equation}
A_{n4j}^{(2)}=O\left(n^{1/2}\log n\left(\phi^{-1}\left(\frac{k_0}{n}\right) ^{\alpha}+\sqrt{\frac{\psi_{S_{\mathcal{F}}}(\frac{\log n}{n})}{k_0}}\right)\right)= o(\sqrt{n}) \ a.s.  \label{A42}
\end{equation}
and
\begin{equation}
A_{n5j}^{(2)}=O\left(n^{1/2}\log n\left(\phi^{-1}\left(\frac{k_1}{n}\right) ^{\alpha}+\sqrt{\frac{\psi_{S_{\mathcal{F}}}(\frac{\log n}{n})}{k_1}}\right)\right) = o(\sqrt{n}) \ a.s.,  \label{A52}
\end{equation}
respectively. 
Finally, taking (\ref{O-eta}) and (\ref{c1})-(\ref{c5}) into account, Abel's inequality gives
\begin{equation}
A_{n1j}^{(3)}=O\left(n \frac{\log n}{\sqrt{k_1}} \left(\phi^{-1}\left(\frac{k_0}{n}\right) ^{\alpha}+\sqrt{\frac{\psi_{S_{\mathcal{F}}}(\frac{\log n}{n})}{k_0}}\right)\right) = o(\sqrt{n}) \ a.s.,  \label{A13}
\end{equation}
\begin{equation}
A_{n2j}^{(3)}=O\left(n \frac{\log n}{\sqrt{k_1}}\left(\phi^{-1}\left(\frac{k_1}{n}\right) ^{\alpha}+\sqrt{\frac{\psi_{S_{\mathcal{F}}}(\frac{\log n}{n})}{k_1}}\right)\right) = o(\sqrt{n}) \ a.s.,  \label{A23}
\end{equation}
\begin{equation}
A_{n4j}^{(3)}=O\left(n \frac{\log n}{\sqrt{k_1}}\left(\phi^{-1}\left(\frac{k_0}{n}\right) ^{\alpha}+\sqrt{\frac{\psi_{S_{\mathcal{F}}}(\frac{\log n}{n})}{k_0}}\right)\right) = o(\sqrt{n}) \ a.s.  \label{A43}
\end{equation}
and
\begin{equation}
A_{n5j}^{(3)}=O\left(n \frac{\log n}{\sqrt{k_1}}\left(\phi^{-1}\left(\frac{k_1}{n}\right) ^{\alpha}+\sqrt{\frac{\psi_{S_{\mathcal{F}}}(\frac{\log n}{n})}{k_1}}\right)\right) = o(\sqrt{n}) \ a.s. \label{A53}
\end{equation}

Next we study the term $A_{n3}$. As in the previous proofs, we will obtain the a.s. orders of $A_{n3j}^{(u)}$ ($u=1,2,3$). To do this, we first need the order of
\begin{equation}
c_{i3}=-\sum_{s=1}^n\omega_{k_1}(\mathcal{X}_i,\mathcal{X}_s) (1-\delta_s)X_s^{\top}(\widehat{\pmb{\beta}}^C-\pmb{\beta}).\nonumber
\end{equation}
On the one hand, from the facts that $g_l$ ($1 \leq l \leq p$) is bounded and
\begin{equation}
\sum_{s=1}^n|\omega_{k_1}(\mathcal{X}_i,\mathcal{X}_s) (1-\delta_s) |\leq 1, \label{sum-w}
\end{equation}
we obtain that 
\begin{equation}
\max_{1\leq i \leq n}\left|\sum_{s=1}^n\omega_{k_1}(\mathcal{X}_i,\mathcal{X}_s) (1-\delta_s)g_l(\mathcal{X}_s)\right|=O(1).
\end{equation}
On the other hand, by arguments similar to those used to obtain (\ref{O-eta}), we have that 
\begin{equation}
\max_{1\leq i \leq n}\left|\sum_{s=1}^n\omega_{k_1}(\mathcal{X}_i,\mathcal{X}_s) (1-\delta_s)\eta_{sl}\right|=O\left(\frac{\log n}{\sqrt{k_1}}\right) \ a.s.
\end{equation}
These two results, together with Lemma \ref{lemma3.7} (see (\ref{betaC-log})) and the fact that $\log^2n=O(k_1)$, give
\begin{equation}
\max_{1\leq i \leq n}|c_{i3}|= O\left(\sqrt{\frac{\log\log n}{n}}\right) \ a.s. \label{c3}
\end{equation}
Using similar arguments to those used in the study of $A_{nrj}^{(u)}$ ($r=1,2,4,5; \ u=1,2,3$), and taking (\ref{c3}) into account, we obtain that 
\begin{equation}
A_{n3j}^{(1)}=O\left(\sqrt{n\log\log n}\left(\phi^{-1}\left(\frac{k_1}{n}\right) ^{\alpha}+\sqrt{\frac{\psi_{S_{\mathcal{F}}}(\frac{\log n}{n})}{k_1}}\right)\right)=o(\sqrt{n}) \ a.s, \label{A31}
\end{equation}

\begin{equation}
A_{n3j}^{(2)}=O\left(\log n \sqrt{\log\log n}\right)=o(\sqrt{n}) \ a.s. \label{A32}
\end{equation}
and
\begin{equation}
A_{n3j}^{(3)}=O\left(n^{1/2}k_1^{-1/2}\log n \sqrt{\log\log n}\right)=o(\sqrt{n}) \ a.s. \label{A33}
\end{equation}
Note that from (\ref{A1-7}), (\ref{A11})-(\ref{A53}) and (\ref{A31})-(\ref{A33}) it follows that
\begin{equation}
A_{n1}+A_{n2}+A_{n3}+A_{n4}+A_{n5}=o(\sqrt{n}) \ a.s. \label{sum_A12345}
\end{equation}

Now we focus on $A_{n6}$. We have that
\begin{equation}
A_{n6}=A_{n6}^{(1\ast)}+A_{n6}^{(2\ast)}+A_{n6}^{(3\ast)}+A_{n4}^{(4\ast)}, \label{sum_A6}
\end{equation}
where
$$
A_{n6}^{(1\ast)}=\sum_{i=1}^n({X}_i-\pmb{g}(\mathcal{X}_i))(1-\delta_i)(X_i-\widehat{\pmb{g}}^C(\mathcal{X}_i))^{\top}(\widehat{\pmb{\beta}}^C-\pmb{\beta}), 
$$

$$
A_{n6}^{(2\ast)}=\sum_{i=1}^n(\pmb{g}(\mathcal{X}_i)-\widehat{\pmb{g}}(\mathcal{X}_i))(1-\delta_i)X_i^{\top}(\widehat{\pmb{\beta}}^C-\pmb{\beta}), 
$$

$$
A_{n6}^{(3\ast)}=\sum_{i=1}^n(\pmb{g}(\mathcal{X}_i)-\widehat{\pmb{g}}(\mathcal{X}_i))(1-\delta_i)(\widehat{m}^C(\mathcal{X}_i)-m(\mathcal{X}_i))
$$
and
$$
A_{n6}^{(4\ast)}=-\sum_{i=1}^n(\pmb{g}(\mathcal{X}_i)-\widehat{\pmb{g}}(\mathcal{X}_i))(1-\delta_i)(\widehat{m}^{C\ast}(\mathcal{X}_i)-m(\mathcal{X}_i)).
$$
In order to obtain the asymptotic normality of the estimator $\widehat{\pmb{\beta}}^C$, in Lemma \ref{lemma3.7} (see (\ref{betaC-norm})) it was proven that 
\begin{equation}
\sqrt{n}(\widehat{\pmb{\beta}}^C-\pmb{\beta})=\pmb{\Sigma}_{0}^{-1}n^{-1/2}\left(\sum_{i=1}^{n}\eta_i^{C}\delta_i\varepsilon_i \right) + o(1) \ a.s \label{desc-betaC}
\end{equation}
By combining (\ref{desc-betaC}) and the strong law of large numbers, we get
\begin{equation}
A_{n6}^{(1\ast)}=\mathbb{E}((1-\Delta(X_1,\mathcal{X}_{1}))\eta_1\eta_1^{C {\top}})\pmb{\Sigma}_{0}^{-1}\sum_{i=1}^{n}\eta_i^{C}\delta_i\varepsilon_i  + o(\sqrt{n}) \ a.s. \label{A61ast}
\end{equation}
From (\ref{O_ghat}) and Lemma \ref{lemma3.7} (see (\ref{betaC-log})), together with the strong law of large numbers, we get
\begin{equation}
A_{n6}^{(2\ast)}= O\left(\phi^{-1}\left(\frac{k_1}{n}\right) ^{\alpha}+\sqrt{\frac{\psi_{S_{\mathcal{F}}}(\frac{\log n}{n})}{k_1}}\right)O\left(n^{-1/2}(\log\log n)^{1/2}\right)O(n)=o(\sqrt{n}) \ a.s. \label{A62ast}
\end{equation}
In addition, Abel's inequality togeher with (\ref{O_ghat}) and (\ref{O_mtilde}) gives
\begin{equation}
A_{n6}^{(3\ast)}= nO\left(\phi^{-1}\left(\frac{k_1}{n}\right) ^{\alpha}+\sqrt{\frac{\psi_{S_{\mathcal{F}}}(\frac{\log n}{n})}{k_1}}\right)\left(\phi^{-1}\left(\frac{k_0}{n}\right) ^{\alpha}+\sqrt{\frac{\psi_{S_{\mathcal{F}}}(\frac{\log n}{n})}{k_0}}\right)=o(\sqrt{n}) \ a.s., \label{A63ast}
\end{equation}
while if one uses (\ref{O_mtildeCast}) instead of (\ref{O_mtilde}), it is obtained that
\begin{equation}
A_{n6}^{(4\ast)}= nO\left(\phi^{-1}\left(\frac{k_1}{n}\right) ^{\alpha}+\sqrt{\frac{\psi_{S_{\mathcal{F}}}(\frac{\log n}{n})}{k_1}}\right)\left(\phi^{-1}\left(\frac{k_0}{n}\right) ^{\alpha}+\sqrt{\frac{\psi_{S_{\mathcal{F}}}(\frac{\log n}{n})}{k_0}}\right)=o(\sqrt{n}) \ a.s. \label{A64ast}
\end{equation}

Finally, we study the term $A_{n7}$. From (\ref{A1-7}), we have that each of its $j$-th components, $A_{n7j}$ ($j=1,\ldots,p$), can be decomposed as
\begin{equation}
A_{n7j}=\sum_{i=1}^n\widetilde{G}_{ij}\delta_i \varepsilon_i + \sum_{i=1}^n\eta_{ij}\delta_i \varepsilon_i + \sum_{i=1}^n (\sum_{s=1}^n \omega_{k_1}(\mathcal{X}_i,\mathcal{X}_s)\eta_{sj})\delta_i \varepsilon_i=A_{n7j}^{(1)}+A_{n7j}^{(2)}+A_{n7j}^{(3)}.\label{A7}
\end{equation}
On the one hand, considering $a_{ij}=\widetilde{G}_{ij}$ and $b_n=na_n$ in Lemma \ref{lemma4}, and taking (\ref{O_gtilde}) into account, we obtain that
\begin{equation}
A_{n7j}^{(1)}=O\left(n^{1/2}\log n \left(\phi^{-1}\left(\frac{k_1}{n}\right) ^{\alpha}+\sqrt{\frac{\psi_{S_{\mathcal{F}}}(\frac{\log n}{n})}{k_1}}\right)\right)=o(\sqrt{n}) \ a.s.  \label{A71}
\end{equation}
On the other hand, considering $a_{ij}=\sum_{s=1}^n \omega_{k_1}(\mathcal{X}_i,\mathcal{X}_s)\eta_{sj}$ and $b_n=na_n$ in Lemma \ref{lemma4}, and taking (\ref{O-eta}) into account, we obtain that
\begin{equation}
A_{n7j}^{(3)}=O\left(n^{1/2}\log n \frac{\log n}{\sqrt{k_1}} \right)=o(\sqrt{n}) \ a.s.\label{A73}
\end{equation}
Now, from (\ref{sum_A12345}), (\ref{sum_A6}) and (\ref{A61ast})-(\ref{A73}), we have that
\begin{equation}
A_n=\sum_{i=1}^n\eta_{i}\delta_i \varepsilon_i + \mathbb{E}((1-\Delta(X_1,\mathcal{X}_{1}))\eta_1\eta_1^{C {\top}})\pmb{\Sigma}_{0}^{-1}\sum_{i=1}^{n}\eta_i^{C}\delta_i\varepsilon_i + o(\sqrt{n}) \ a.s. \label{An}
\end{equation}
Finally, 
(\ref{betaI-beta}) and (\ref{An}), together with both Lemma \ref{lemma3} and a central limit theorem, prove Theorem \ref{theorem1}(i). $\Box$

\noindent \textit{Proof of Theorem \ref{theorem1}(ii).} From (\ref{betaI-beta}), (\ref{An}) and Lemma \ref{lemma3}, we have that
\begin{equation}
\widehat{\pmb{\beta}}^I - \pmb{\beta}=\pmb{\Sigma}_{1}^{-1} n^{-1}\sum_{i=1}^n\eta_{i}\delta_i \varepsilon_i + \pmb{\Gamma}n^{-1}\sum_{i=1}^{n}\eta_i^{C}\delta_i\varepsilon_i + o(n^{-1/2}) =B_{n1}+B_{n2} +o(n^{-1/2})\ a.s., \label{log-it}
\end{equation}
where
$$\pmb{\Gamma}=\pmb{\Sigma}_{1}^{-1}\mathbb{E}((1-\Delta(X_1,\mathcal{X}_{1}))\eta_1\eta_1^{C {\top}})\pmb{\Sigma}_{0}^{-1}.$$
Firstly, we obtain two auxiliary results that will play a main role in the proof of this Theorem. Let $\pmb{\Sigma}_{1j}^{-1}$ and $\pmb{\Gamma}_{j}$ denote the $j$-th rows of the matrices $\pmb{\Sigma}_{1}^{-1}$ and $\pmb{\Gamma}$, respectively. Considering $V_i=\pmb{\Sigma}_{1j}^{-1\top}\eta_{i}\delta_i \varepsilon_i/Var(\pmb{\Sigma}_{1j}^{-1\top}\eta_{i}\delta_i \varepsilon_i)^{1/2}$ in Lemma \ref{LemaStout} 

we obtain that
\begin{eqnarray}
	\limsup_{n\rightarrow \infty}\left(\frac{n}{2\log \log n}\right)^{1/2}\left|n^{-1}\sum_{i=1}^n  \pmb{\Sigma}_{1j}^{-1\top}\eta_{i}\delta_i \varepsilon_i\right|&=&\limsup_{n\rightarrow \infty}\left(\frac{1}{2n\log \log n}\right)^{1/2}\left|\sum_{i=1}^n  \pmb{\Sigma}_{1j}^{-1\top}\eta_{i}\delta_i \varepsilon_i\right| \nonumber \\ &=&	Var(\pmb{\Sigma}_{1j}^{-1\top}\eta_{1}\delta_1 \varepsilon_1)^{1/2} \ a.s. \label{LL1}
\end{eqnarray}
In the same way, if we consider $V_i=\pmb{\Gamma}_{j}^{\top}\eta_{i}\delta_i^C \varepsilon_i/Var(\pmb{\Gamma}_{j}^{\top}\eta_{i}^C\delta_i\varepsilon_i)^{1/2}$ in Lemma \ref{LemaStout}, we obtain
\begin{equation}
	\limsup_{n\rightarrow \infty}\left(\frac{n}{2\log \log n}\right)^{1/2}\left|n^{-1}\sum_{i=1}^n  \pmb{\Gamma}_{j}^{\top}\eta_{i}^C\delta_i \varepsilon_i\right|=Var(\pmb{\Gamma}_{j}^{\top}\eta_{1}^C\delta_1 \varepsilon_1)^{1/2} \ a.s. \label{LL2}
\end{equation}
Next, we focus on the proof of Theorem \ref{theorem1}(ii). We have that  
\begin{eqnarray}
\widehat{m}^I(\chi)-m(\chi)&=&
\left(\sum_{j=1}^n\omega_{k_1}(\chi,\mathcal{X}_j) Y_j^{I\ast}-f(\chi)   \right) + \left(\sum_{j=1}^n\omega_{k_1}(\chi,\mathcal{X}_j) (1-\delta_j)X_j^{\top}(\widehat{\pmb{\beta}}^C-\pmb{\beta})\right) \nonumber \\
&+& \left(\sum_{j=1}^n\omega_{k_1}(\chi,\mathcal{X}_j) (1-\delta_j)(\widehat{m}^C(\mathcal{X}_j)-m(\mathcal{X}_j))\right) \nonumber\\
&-& \left(\widehat{\pmb{g}}(\chi)-\pmb{g}(\chi) \right)^{\top}\left(\widehat{\pmb{\beta}}^I-\pmb{\beta} \right) -\pmb{g}^{\top}\left( \widehat{\pmb{\beta}}^I-\pmb{\beta}\right)
- \left(\widehat{\pmb{g}}(\chi)-\pmb{g}(\chi) \right)^{\top}\pmb{\beta} \nonumber\\
&=& B_{n1}(\chi)+B_{n2}(\chi)+B_{n3}(\chi)+B_{n4}(\chi)+B_{n5}(\chi)+B_{n6}(\chi). \label{sum_Bn}
\end{eqnarray}
Note that $$B_{n1}(\chi)+B_{n2}(\chi)+B_{n3}(\chi)=\sum_{j=1}^n\omega_{k_1}(\chi,\mathcal{X}_j) Y_j^{I}-f(\chi).$$
From (\ref{O_ftildeast}) we have that
\begin{eqnarray} 
\sup_{\chi\in S_{\mathcal{F}}}|B_{n1}(\chi)|=O\left(\phi^{-1}\left(\frac{k_1}{n}\right) ^{\alpha}+\sqrt{\frac{\psi_{S_{\mathcal{F}}}(\frac{\log n}{n})}{k_1}}\right) \ a.s.,\label{B1}
\end{eqnarray}
while (\ref{c3}) gives
\begin{eqnarray}
\sup_{\chi\in S_{\mathcal{F}}}|B_{n2}(\chi)|=O\left(\sqrt{\frac{\log\log n}{n}}\right) \ a.s.\label{B2}
\end{eqnarray} 
From (\ref{O_mtilde}) together with (\ref{sum-w}), we have that
\begin{eqnarray}
\sup_{\chi\in S_{\mathcal{F}}}=|B_{n3}(\chi)|=O\left(\phi^{-1}\left(\frac{k_0}{n}\right) ^{\alpha}+\sqrt{\frac{\psi_{S_{\mathcal{F}}}(\frac{\log n}{n})}{k_0}}\right) \ a.s. \label{B3}
\end{eqnarray}
In addition, from (\ref{O_ghat}) and (\ref{log-it})-(\ref{LL2}) we have that
\begin{eqnarray}
\sup_{\chi\in S_{\mathcal{F}}}|B_{n4}(\chi)|=O\left(\phi^{-1}\left(\frac{k_1}{n}\right) ^{\alpha}+\sqrt{\frac{\psi_{S_{\mathcal{F}}}(\frac{\log n}{n})}{k_1}}\right) \ a.s.,\label{B4}
\end{eqnarray}
while (\ref{log-it})-(\ref{LL2}) give 
\begin{eqnarray}
\sup_{\chi\in S_{\mathcal{F}}}|B_{n5}(\chi)|=O\left(\sqrt{\frac{\log\log n}{n}}\right) \ a.s.\label{B5}
\end{eqnarray}
Finally, from (\ref{O_ghat}) we have that 
\begin{eqnarray}
\sup_{\chi\in S_{\mathcal{F}}}|B_{n6}(\chi)|=O\left(\phi^{-1}\left(\frac{k_1}{n}\right) ^{\alpha}+\sqrt{\frac{\psi_{S_{\mathcal{F}}}(\frac{\log n}{n})}{k_1}}\right) \ a.s. \label{B6}
\end{eqnarray}
The claimed result follows from (\ref{sum_Bn})-(\ref{B6}). $\Box$

\subsection{Proof of Theorem \ref{theorem2}}\label{app3}

\textit{Proof of Theorem \ref{theorem2}(i).} We can write

\begin{equation}
n^{1/2}(\widehat{\pmb{\beta}}^{R}-\pmb{\beta})=\left(n^{-1}\widetilde{\mathbf{X}}_1^{\top}\widetilde{\mathbf{X}}_1\right)^{-1} n^{-1/2} C_n, \label{betaR-beta}
\end{equation}
where 
$$C_n=\widetilde{\mathbf{X}}_1^{\top}\left(\widetilde{\mathbf{Y^R}}_1-\widetilde{\mathbf{X}}_1\pmb{\beta}\right).$$
Let us denote $Y_i^{R\ast}=X_i^{\top}\pmb{\beta}+m(\mathcal{X}_i)$ and let us remember the notation $\widehat{m}^{C\ast}(\mathcal{X}_i) =\widehat{f}^C(\mathcal{X}_i)-\widehat{\pmb{g}}^C(\mathcal{X}_i)^{\top}\pmb{\beta}$. We have that
\begin{equation}
C_n=\sum_{r=1}^6C_{nr}, \label{C}
\end{equation}
where
$$C_{n1}=\sum_{i=1}^n\widetilde{X}_i(\widehat{m}^{C\ast}(\mathcal{X}_i)-m(\mathcal{X}_i)),
$$
$$C_{n2}=-\sum_{i=1}^n\widetilde{X}_i\left(\sum_{j=1}^n\omega_{k_1}(\mathcal{X}_i,\mathcal{X}_j) Y_j^{R\ast}-f(\mathcal{X}_i)   \right),
$$
$$C_{n3}=-\sum_{i=1}^n\widetilde{X}_i\left(\sum_{j=1}^n\omega_{k_1}(\mathcal{X}_i,\mathcal{X}_j) X_j^{\top}(\widehat{\pmb{\beta}}^C-\pmb{\beta})\right), 
$$
$$C_{n4}=-\sum_{i=1}^n\widetilde{X}_i\left(\sum_{j=1}^n\omega_{k_1}(\mathcal{X}_i,\mathcal{X}_j)(\widehat{m}^C(\mathcal{X}_j)-m(\mathcal{X}_j))\right),
$$
$$C_{n5}=\sum_{i=1}^n\widetilde{X}_i\left(\widehat{\pmb{g}}(\mathcal{X}_i)-\pmb{g}(\mathcal{X}_i)\right)^{\top}\pmb{\beta}
$$
and
$$C_{n6}=\sum_{i=1}^n\widetilde{X}_i(X_i-\widehat{\pmb{g}}^C(\mathcal{X}_i))^{\top}(\widehat{\pmb{\beta}}^C-\pmb{\beta}).
$$
On the one hand, noting that the only difference between $C_{nr}$ and $A_{nr}$ ($r=1,3,4$) is the presence of the factor $1-\delta$ in $A_{nr}$, from similar arguments to those used to obtain that $A_{nr}=o(\sqrt{n}) \ a.s.$ (see (\ref{sum_A12345})) we have that $C_{nr}=o(\sqrt{n}) \ a.s.$ holds. On the other hand, because of $C_{n5}=A_{n5}$, (\ref{sum_A12345}) also gives $C_{n5}=o(\sqrt{n}) \ a.s.$ Focusing now on $C_{n2}$, its only difference with $A_{n2}$ is that in $C_{n2}$ the variable $Y_j^{R\ast}$ is used instead of $Y_j^{I\ast}$. Taking into account that, under the MAR condition, $\mathbb{E}\left(Y_i^{R\ast}|\mathcal{X}_i\right)=\mathbb{E}\left(Y_i^{I\ast}|\mathcal{X}_i\right)=f(\mathcal{X}_i)$ holds, similar arguments to those used to obtain that $A_{n2}=o(\sqrt{n}) \ a.s.$ (see (\ref{sum_A12345})) give that $C_{n2}=o(\sqrt{n}) \ a.s.$ holds. In short, we have that
\begin{equation}
C_{n1}+C_{n2}+C_{n3}+C_{n4}+C_{n5}=o(\sqrt{n}) \ a.s. \label{sum_C12345}
\end{equation}
Finally, we consider the term $C_{n6},$ which can be written as
\begin{equation}
C_{n6}=C_{n6}^{(1\ast)}+C_{n6}^{(2\ast)}+C_{n6}^{(3\ast)}+C_{n4}^{(4\ast)} \ a.s., \label{sum_C6}
\end{equation}
where
$$
C_{n6}^{(1\ast)}=\sum_{i=1}^n({X}_i-\pmb{g}(\mathcal{X}_i))(X_i-\widehat{\pmb{g}}^C(\mathcal{X}_i))^{\top}(\widehat{\pmb{\beta}}^C-\pmb{\beta}), 
$$

$$
C_{n6}^{(2\ast)}=\sum_{i=1}^n(\pmb{g}(\mathcal{X}_i)-\widehat{\pmb{g}}(\mathcal{X}_i))X_i^{\top}(\widehat{\pmb{\beta}}^C-\pmb{\beta}), 
$$

$$
C_{n6}^{(3\ast)}=\sum_{i=1}^n(\pmb{g}(\mathcal{X}_i)-\widehat{\pmb{g}}(\mathcal{X}_i))(\widehat{m}^C(\mathcal{X}_i)-m(\mathcal{X}_i))
$$
and
$$
C_{n6}^{(4\ast)}=-\sum_{i=1}^n(\pmb{g}(\mathcal{X}_i)-\widehat{\pmb{g}}(\mathcal{X}_i))(\widehat{m}^{C\ast}(\mathcal{X}_i)-m(\mathcal{X}_i)).
$$
Again, the only difference between $C_{n6}^{(r\ast)}$ and $A_{n6}^{(r\ast)}$ is the presence of the factor $1-\delta$ in $A_{n6}^{(r\ast)}$ ($r=1,2,3,4$). In the case of $A_{n6}^{(1\ast)}$, one can see in (\ref{A61ast}) that $1-\delta$ only influences by means of the term $1-\Delta(X_1,\mathcal{X}_{1}).$ Taking this fact into account, from similar arguments to those used to obtain (\ref{A61ast}), we have that
\begin{equation}
C_{n6}^{(1\ast)}=\pmb{\Sigma}_{3}\pmb{\Sigma}_{0}^{-1}\sum_{i=1}^{n}\eta_i^{C}\delta_i\varepsilon_i  + o(\sqrt{n}) \ a.s. \label{C61ast}
\end{equation}
In addition, it is easy to see that $1-\delta$ plays an irrelevant role in obtaining the order of $A_{n6}^{(r\ast)}$ ($r=2,3,4$); see (\ref{A62ast})-(\ref{A64ast}) for details. Therefore, $C_{n6}^{(2\ast)}+C_{n6}^{(3\ast)}+C_{n6}^{(4\ast)}=o(\sqrt{n}) \ a.s.$ From this fact together with (\ref{C})-(\ref{C61ast}) we have that
\begin{equation}
C_{n}=\pmb{\Sigma}_{3}\pmb{\Sigma}_{0}^{-1}\sum_{i=1}^{n}\eta_i^{C}\delta_i\varepsilon_i  + o(\sqrt{n}) \ a.s. \label{Cn-desc}
\end{equation}
Finally, (\ref{betaR-beta}) and (\ref{Cn-desc}), together with both Lemma \ref{lemma3} and a central limit theorem, prove Theorem \ref{theorem2}(i). $\Box$

\noindent \textit{Proof of Theorem \ref{theorem2}(ii).} It can be obtained using similar arguments to those considered to prove Theorem \ref{theorem1}(ii). Therefore we omit it here. $\Box$

\subsection{Proof of Theorem \ref{theorem3}}\label{app4}

\textit{Proof of Theorem \ref{theorem3}(i).} We have that

\begin{equation}
n^{1/2}(\widehat{\pmb{\beta}}^{IP}-\pmb{\beta})=\left(n^{-1}\widetilde{\mathbf{X}}_1^{\top}\widetilde{\mathbf{X}}_1\right)^{-1} n^{-1/2} D_n, \label{betaIP-beta}
\end{equation}
where 
$$D_n=\widetilde{\mathbf{X}}_1^{\top}\left(\widetilde{\mathbf{Y^{IP}}}_1-\widetilde{\mathbf{X}}_1\pmb{\beta}\right).$$
Let us denote
\begin{equation}
Y_i^{IP\ast}=\frac{\delta_i}{\widehat{\Delta}_1(\mathcal{X}_i)}Y_i+(1-\frac{\delta_i}{\widehat{\Delta}_1(\mathcal{X}_i)})\left(X_i^{\top}\pmb{\beta}+m(\mathcal{X}_i)  \right).
\end{equation}
Then, in a similar way as in the decomposition (\ref{A}), we can write 
\begin{equation}
D_n=\sum_{r=1}^7D_{nr}, \label{D}
\end{equation}
where
$$D_{n1}=\sum_{i=1}^n\widetilde{X}_i(1-\frac{\delta_i}{\widehat{\Delta}_1(\mathcal{X}_i)})(\widehat{m}^{C\ast}(\mathcal{X}_i)-m(\mathcal{X}_i)),
$$

$$D_{n2}=-\sum_{i=1}^n\widetilde{X}_i\left(\sum_{j=1}^n\omega_{k_1}(\mathcal{X}_i,\mathcal{X}_j) Y_j^{IP\ast}-f(\mathcal{X}_i)   \right),
$$

$$D_{n3}=-\sum_{i=1}^n\widetilde{X}_i\left(\sum_{j=1}^n\omega_{k_1}(\mathcal{X}_i,\mathcal{X}_j) (1-\frac{\delta_j}{\widehat{\Delta}_1(\mathcal{X}_j)})X_j^{\top}(\widehat{\pmb{\beta}}^C-\pmb{\beta})\right), 
$$

$$D_{n4}=-\sum_{i=1}^n\widetilde{X}_i\left(\sum_{j=1}^n\omega_{k_1}(\mathcal{X}_i,\mathcal{X}_j) (1-\frac{\delta_j}{\widehat{\Delta}_1(\mathcal{X}_j)})(\widehat{m}^C(\mathcal{X}_j)-m(\mathcal{X}_j))\right),
$$

$$D_{n5}=\sum_{i=1}^n\widetilde{X}_i\left(\widehat{\pmb{g}}(\mathcal{X}_i)-\pmb{g}(\mathcal{X}_i)\right)^{\top}\pmb{\beta},
$$

$$D_{n6}=\sum_{i=1}^n\widetilde{X}_i(1-\frac{\delta_i}{\widehat{\Delta}_1(\mathcal{X}_i)})(X_i-\widehat{\pmb{g}}^C(\mathcal{X}_i))^{\top}(\widehat{\pmb{\beta}}^C-\pmb{\beta}) 
$$
and
$$D_{n7}=\sum_{i=1}^n\widetilde{X}_i\frac{\delta_i}{\widehat{\Delta}_1(\mathcal{X}_i)}\varepsilon_i.
$$
Let $D'_{nr}$ ($r=1,\ldots,7$) denote the term $D_{nr}$ when $\widehat{\Delta}_1(\chi)$ is replaced by $\Delta_1(\chi),$ and let us write $D_{nr}=D'_{nr}+d_{nr}.$ Then, taking into account the fact that $\Delta_1(\chi)>C>0 \ \forall \chi \in S_{\mathcal{F}}$ (see assumption (A7)), from similar arguments to those used in the proof of Theorem \ref{theorem1}(i) (Section \ref{app2}) to obtain (\ref{An}) one obtains
\begin{equation}
\sum_{r=1}^7D'_{nr}=\sum_{i=1}^n\eta_{i}\frac{\delta_i \varepsilon_i}{\Delta_1(\mathcal{X}_i)} + \mathbb{E}((1-\frac{\delta_1}{\Delta_1(\mathcal{X}_1)})\eta_1\eta_1^{C {\top}})\pmb{\Sigma}_{0}^{-1}\sum_{i=1}^{n}\eta_i^{C}\delta_i\varepsilon_i + o(\sqrt{n}) \ a.s. \label{Dnprima}
\end{equation}
In addition, taking into account the almost sure uniform rate of convergence of $\widehat{\Delta}_1(\chi)$ (which can be obtained from Lemma \ref{lemma1} if $g_j(\chi),$ $X_{ij},$ $k_1$, and assumptions (A2)(iv), (A3) and (A5) are replaced by $\Delta_1(\chi),$ $\delta_i,$ $k_2$ and assumptions (A10), (A8) and (A9), respectively), one obtains that 
\begin{equation}
		\sup_{i\in \{1,\ldots,n\}}\left|\frac{\delta_i}{\widehat{\Delta}_1(\mathcal{X}_i)}-\frac{\delta_i}{\Delta_1(\mathcal{X}_i)}\right|= O\left(\phi^{-1}\left(\frac{k_2}{n}\right) ^{\alpha}+\sqrt{\frac{\psi_{S_{\mathcal{F}}}(\frac{\log n}{n})}{k_2}}\right) \ a.s. \nonumber
	\end{equation}
This fact together with similar arguments to those used in Section \ref{app2} to obtain orders and asymptotic representations for $A_{nr}$ ($r=1,\ldots,7$)	give
\begin{equation}
\sum_{r=1}^7d_{nr}=o(\sqrt{n}) \ a.s. \label{dn}
\end{equation}	
Finally, 
(\ref{betaIP-beta}), (\ref{D}), (\ref{Dnprima}) and (\ref{dn}), together with both Lemma \ref{lemma3} and a central limit theorem, prove Theorem \ref{theorem3}(i). $\Box$

\noindent \textit{Proof of Theorem \ref{theorem3}(ii).} It can be obtained using similar arguments to those considered to prove Theorem \ref{theorem1}(ii). Therefore we omit it here. $\Box$

\begin{thebibliography}{99}
	\bibitem[Aneiros et al.(2015)]{ane15} Aneiros G, Ferraty F, Vieu P (2015) Variable selection in partial linear regression with functional covariate. Statistics 49:1322--1347
	
	\bibitem[Aneiros et al.(2022)Aneiros, Horová, Hušková and Vieu]{anehhv22} Aneiros G, Horová I, Hušková M, Vieu P (2022) On functional data analysis and related topics. J Multivariate Anal 189:104861
	
	\bibitem[Aneiros-Pérez and Vieu(2006)]{anev06} Aneiros-Pérez G, Vieu P (2006) Semi-functional partial linear regression. Statist Probab Lett 76:1102--1110
	
	\bibitem[Aneiros-Pérez and Vieu(2008)]{anev08} Aneiros-Pérez G, Vieu P (2008) Nonparametric time series prediction: A semi-functional partial linear modeling. J Multivariate Anal 99:834--857
		
	\bibitem[Boente and Vahnovan(2017)]{boev17} Boente G, Vahnovan A (2017) Robust estimators in semi-functional partial linear regression models. J Multivariate Anal 154:59--84 


	\bibitem[Collomb(1979)]{collomb} Collomb G (1979) Estimation de la régression par la méthode des k points les plus proches: propriétés de convergence ponctuelle. C R Acad Sci Paris 289:245--247
	
	\bibitem[Crambes and Henchiri(2019)]{crah19}Crambes C, Henchiri Y (2018) Regression imputation in the functional linear model with missing values in the response. J Statist Plann Inference 201:103--119
	

	\bibitem[Ding et al.(2018)Ding, Lu, Zhang and Zhang]{din18} Ding H, Lu Z, Zhang J, Zhang R (2018) Semi-functional partial linear quantile regression. Statist Probab Lett 142:92--101
	
		\bibitem[Febrero-Bande et al.(2019)Febrero-Bande, Galeano and González-Manteiga]{febg19} Febrero-Bande M, Galeano P, González-Manteiga W (2019) Estimation, imputation and prediction for the functional linear model with scalar response with responses missing at random. Comput Statist Data Anal 131:91--103.
	
	
	\bibitem[Ferraty and Vieu(2006)]{ferv06} Ferraty F, Vieu P (2006) Nonparametric Functional Data Analysis: Theory and practice. Springer, New York
	
	
	
		\bibitem[Ferraty et al.(2013)Ferraty, Sued and Vieu]{fers13} Ferraty F, Sued M, Vieu P (2013) Mean estimation with data missing at random for functional covariables. Statistics 47:688--706.
	
	
		\bibitem[Horváth and Kokoszka(2012)]{hork12} Horváth L, Kokoszka P (2012)  Inference for Functional Data with Applications. Springer.

\bibitem[Hsing and Eubank(2015)]{hsie15} Hsing T, Eubank R (2015) Theoretical Foundations of Functional Data Analysis, With an Introduction to Linear Operators. Wiley.		
		
	\bibitem[Kara-Zaitri et al.(2017)Kara-Zaitri, Laksaci, Rachdi and Vieu]{kara_JMVA} Kara-Zaitri L, Laksaci A, Rachdi M, Vieu P (2017) Data-driven kNN estimation in nonparametric functional data analysis. J Multivariate Anal 153:176--188
	
	\bibitem[Kokoszka et al(2017)Kokoszka, Oja, Park and Sangalli]{koko17}Kokoszka P, Oja H, Park B, Sangalli L (2017). Special issue on functional data analysis. Econom Stat 1: 99--100.

\bibitem[Kokoszka and Reimherr(2017)]{kokr17} Kokoszka P, Reimherr M (2017) Introduction to Functional Data Analysis. CRC Press
	
	\bibitem[Kudraszow and Vieu(2013)]{Kud13} Kudraszow N, Vieu P (2013) Uniform consistency of $k$NN regressors for functional variables. Statist Probab Lett 83:1863--1870
	
	\bibitem[Ling et al.(2019)Ling, Kan, Vieu and Meng]{lin19} Ling N, Kan R, Vieu P, Meng S (2019) Semi-functional partially linear regression model with responses missing at random. Metrika 82:39--70 
	
	\bibitem[Ling et al.(2020)Ling, Aneiros and Vieu]{linav20} Ling N, Aneiros G, Vieu P (2020) kNN estimation in functional partial linear modeling. Statist Papers 61:423--444
	
	
	\bibitem[Little and Rubin(1987)]{litr87} Little RJA, Rubin DB (1987) Statistical Analysis with Missing Data. Wiley, New York
	
	\bibitem[Novo et al.(2019)Novo, Aneiros and Vieu]{novo_2019} Novo S, Aneiros G, Vieu P (2019) Automatic and location-adaptive estimation in functional single-index regression. J Nonparametr Stat 31:364--392
	
	\bibitem[Rachdi et al.(2021)Rachdi, Laksaci, Kaid, Benchiha and Al-Awadhi]{rac21} Rachdi M, Laksaci A, Kaid Z, Benchiha A, Al-Awadhi FA (2021) \textit{k}-Nearest neighbors local linear regression for functional and missing data at random. Statistica Neerlandica 75:42--65
	
	\bibitem[Ramsay and Silverman(2005)]{rams05} Ramsay J, Silverman B (2005) Functional Data Analysis. Springer, New York
	
	\bibitem[Shang(2014)]{sha14} Shang HL (2014) Bayesian bandwidth estimation for a semi-functional partial linear regression model with unknown error density. Comput Statist 29:829--848
	
	\bibitem[Stout(1974)]{stout74} Stout W (1974) Almost Sure Convergence. Academic Press, New York
		
 \bibitem[Vilar et al.(2012)Vilar, Cao and Aneiros]{vil12} Vilar JM, Cao R and Aneiros (2012) Forecasting next-day electricity demand and price using nonparametric functional methods. Int J Electr Power and Energy Syst 39:48--55

\bibitem[Wang(2026)]{wan26} Wang T (2026) Robust semi-functional censored regression. J Multivariate Anal 211:105491

		\bibitem[Wang et al.(2004)Wang, Linton and Hardle]{wan04} Wang Q, Linton O and Hardle H (2004) Semiparametric regression analysis with missing responses at random. J Amer Statist Assoc 99:334--345

	\bibitem[Wang and Sun(2007)]{wan07} Wang Q, Sun Z (2007) Estimation in partially linear models with missing responses at random. J Multivariate Anal 98:1470--1493
	
	
	\bibitem[Zhu et al.(2020)]{zhu20} Zhu H, Zhang R, Zhu G (2020) Estimation and inference in semi-functional partially linear measurement error models. J Syst Sci Complex 33:1179--1199
	
	
	\bibitem[Zhu and Zhao(2019)]{zhuz19} Zhu S, Zhao P (2019) Tests for the linear hypothesis in semi-functional partial linear regression models.	Metrika 82:125--148
	
\end{thebibliography}
\end{document}